\documentclass[11pt,a4paper,logo]{googledeepmind}

\setleftlogo[180pt]{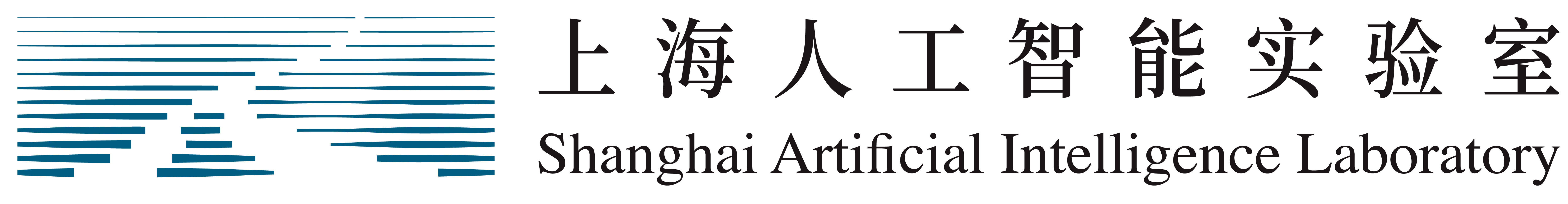} 
\setrightlogo[180pt]{}

\usepackage{algorithm}
\usepackage[noend]{algpseudocode}
\usepackage{amssymb}
\usepackage{placeins}
\usepackage{subcaption}

\usepackage[
    natbib=true,
    backend=biber,
    style=numeric, 
    sorting=none 
]{biblatex}
\AtEveryBibitem{\clearfield{month}}
\AtEveryBibitem{\clearfield{day}}

\usepackage{csquotes}

\newcommand{\symbqa}{X}
\newcommand{\symbPre}{\mathbb{P}}
\newcommand{\symbpre}{p}
\newcommand{\symbhyp}{q}
\newcommand{\symbAutofm}{AF}
\newcommand{\symbPF}{\mathcal{P}}
\newcommand{\symbpf}{\mathfrak{p}}
\newcommand{\symbHF}{\mathfrak{q}}
\newcommand{\symbSymTab}{\mathcal{T}}
\newcommand{\symbSolver}{\mathcal{S}}
\newcommand{\symbgoal}{\mathfrak{g}}
\newcommand{\symbBridge}{\mathcal{B}}
\newcommand{\symbbridge}{\mathfrak{b}}
\newcommand{\symbterm}{\mathfrak{e}}
\newcommand{\symbInst}{\mathcal{I}}
\newcommand{\symbinst}{\mathfrak{i}}
\newcommand{\symbAncestor}{\mathcal{A}}
\newcommand{\symbobligation}{\mathfrak{o}}

\newcommand{\Projectname}{\texttt{Safin-1.5}}
\newcommand{\ReSI}{\mbox{ReSI}}

\newsavebox{\safintitlebox}

\title{%
    \sbox{\safintitlebox}{%
    \begin{minipage}[c]{0.84\linewidth}
        \raggedright
        \fontsize{20}{18}\selectfont
        \Projectname{}: Safety-Native Architecture and Verifiable Reasoning for Recursive Safety Improvement
    \end{minipage}}%
    \raisebox{-\dp\safintitlebox}{%
        \includegraphics[height=\dimexpr\ht\safintitlebox+\dp\safintitlebox\relax]{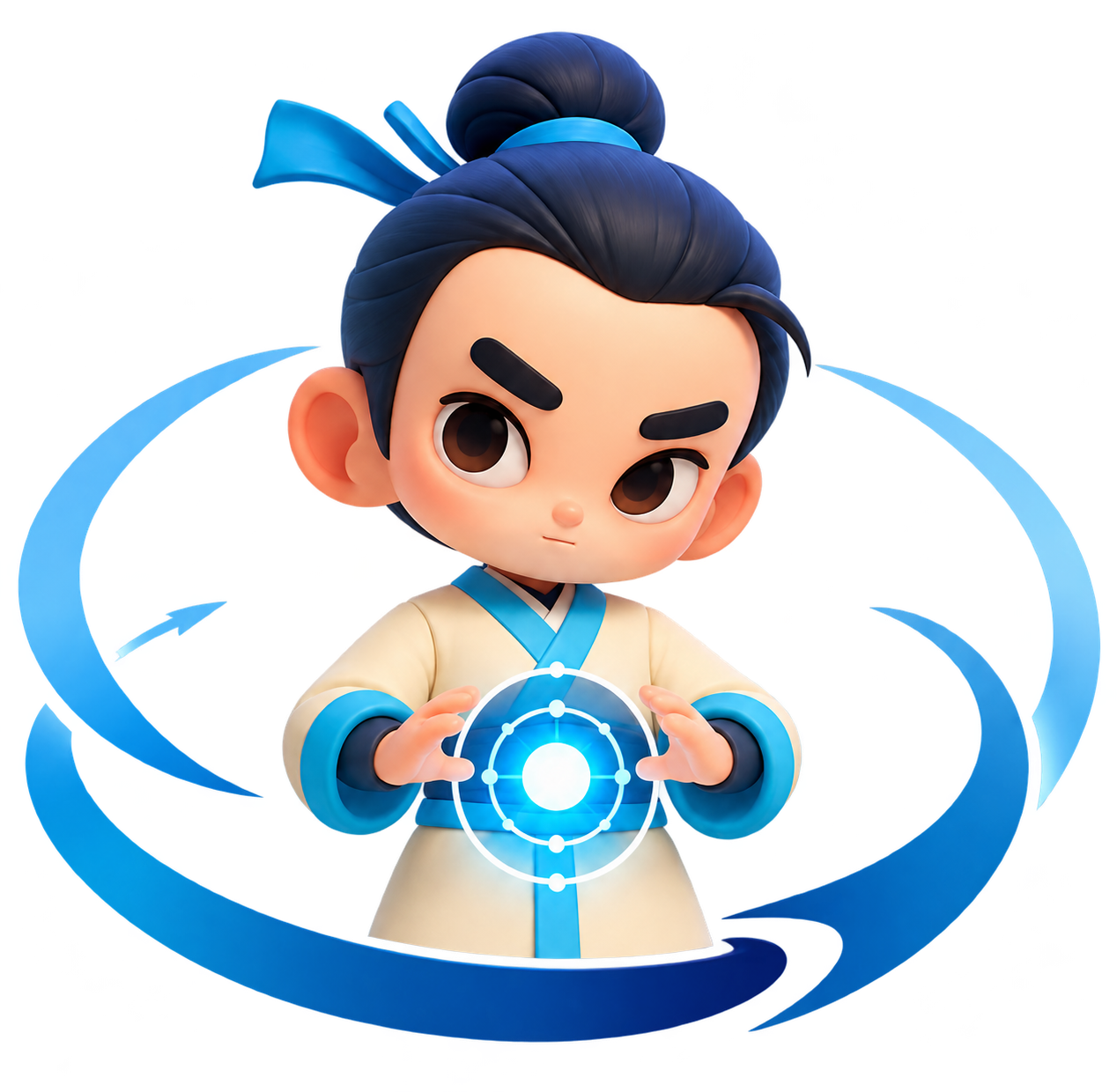}%
    }\hfill%
    \usebox{\safintitlebox}%
}

\author[1]{Safin Team, Shanghai AI Laboratory}

\usepackage{pdflscape}
\usepackage{amsthm}

\newtheorem{theorem}{Theorem}
\usepackage{textcomp}
\usepackage{rotating}
\usepackage{setspace}
\usepackage{soul} 
\usepackage{multicol}
\usepackage{subcaption}
\usepackage{caption}

\sethlcolor{green!14}

\usepackage{array}
\usepackage{multirow}
\usepackage[justification=centering]{subcaption}
\usepackage{amsmath}
\usepackage{siunitx}

\usepackage{enumitem}
\usepackage{float}
\usepackage{seqsplit}
\usepackage{framed}
\usepackage{tikz}
\usepackage{listings}
\usepackage{colortbl}
\usepackage{wrapfig}
\usepackage[most]{tcolorbox}
\usepackage{pdfpages}
\tcbuselibrary{listingsutf8}
\definecolor{mycolor}{RGB}{50,80,150}

\definecolor{ProjectLinkBorder}{HTML}{39C4DC}
\definecolor{ProjectLinkBg}{HTML}{ECFAFD}
\newcommand{\projectlink}[2]{%
    \href{#1}{\tcbox[on line, colback=ProjectLinkBg,
        colframe=ProjectLinkBorder, boxrule=0.55pt, arc=3pt,
        boxsep=0pt, left=7pt, right=7pt, top=3pt, bottom=3pt]{%
        \makebox[97pt][c]{\color{black}\mdseries\fontsize{10}{12}\selectfont\rule[-3pt]{0pt}{13pt}#2}}}%
}

\usepackage{ragged2e}
\usepackage{makecell}
\usepackage{adjustbox}

\usepackage[symbol]{footmisc}
\newcolumntype{Y}{>{\RaggedRight\arraybackslash}X}

\hypersetup{
    colorlinks=true,
    linkcolor=blue, 
    citecolor=blue,  
    filecolor=black,
    urlcolor=blue    
}
\definecolor{AbstractBgColor}{HTML}{F4F7FB}
\definecolor{tabshadow}{RGB}{237,244,249}
\usepackage{tocloft}

\usepackage{etoolbox}

\patchcmd{\maketitle}{\vskip20pt}{%
    \vskip6pt%
    \projectlink{https://huggingface.co/AI45Research/Safin-1.5}{%
        \raisebox{-3pt}{\includegraphics[height=13pt]{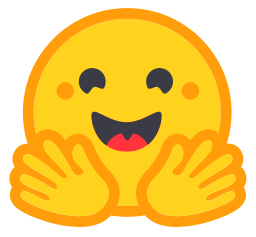}}%
        \hspace{4pt}Hugging Face}\par%
    \vskip12pt%
}{}{\PackageError{Safin}{Model button placement failed}{}} 

\usepackage{arydshln}
\makeatletter
\patchcmd{\@tocline}
    {\hfil}
    {\leaders\hbox{\hfil}\hfil}
    {}{}
\makeatother

\begin{document}
\sloppy
\thispagestyle{firststyle}

\begin{tcolorbox}[
    colback=AbstractBgColor, 
    colframe=AbstractBgColor, 
    arc=5pt,                  
    auto outer arc,
    boxrule=0pt,              
    left=8pt, right=8pt, 
    top=4pt, bottom=4pt,      
    parbox=false,
    width=\textwidth,
    before skip=-800pt,        
    after skip=0pt, 
    enlarge top by=-12pt
]
\begin{abstract}
\vspace{-10pt}

Recursive self-improvement (RSI) enables large language models to advance their capabilities through successive updates, but these updates can also introduce new risks and weaken existing safeguards. In line with the AI-45° Law, we pursue Safety RSI, in which safety develops as an intrinsic capability alongside general model capabilities through continual risk discovery and repair, while preserving utility and limiting over-refusal. We introduce \Projectname{}, a model built on a safety architecture and trained through recursive safety improvement, with an inference-time harness that extends its capabilities to formal reasoning. Our approach comprises three components: i) The safety architecture combines persistent safety states with hazardous-knowledge disentanglement. The former retains and selectively invokes safety knowledge, while the latter relocates targeted hazardous knowledge into removable experts, enabling its removal with minimal impact on general capabilities. ii) Recursive safety improvement establishes a discovery–evaluation–repair loop. The loop identifies new vulnerabilities and uses experimental feedback to develop repair strategies. It accepts updates only when they improve safety while preserving general capabilities and limiting over-refusal. iii) The formal reasoning harness combines the model's language understanding with symbolic inference and verification. Proof-obligation-guided semantic bridging recovers missing relations in formal representations, enabling symbolic solvers to construct checkable proofs under formalized premises and accepted semantic bridges. Experimental results demonstrate that \Projectname{} achieves safety and reasoning performance, outperforming frontier models.



\end{abstract}


\maketitle

\end{tcolorbox}

\begin{figure}[h]
\centering
\includegraphics[width=0.95\textwidth]{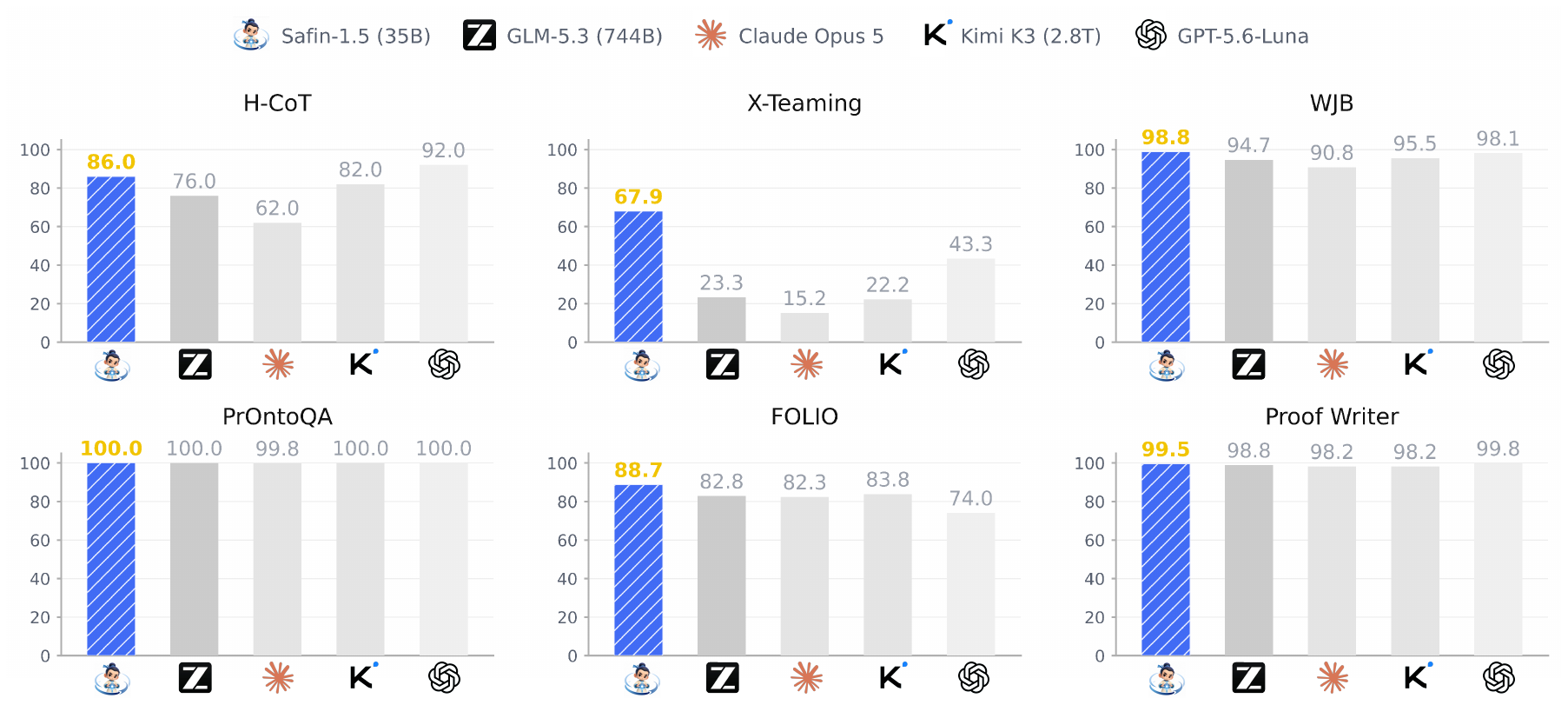}
\caption{Main evaluation results of \Projectname{}~(35B) compared with frontier models across safety and formal reasoning benchmarks. For lower-is-better benchmarks (H-CoT, X-Teaming, and WJB), results are reported as \textbf{100 - metric}, so higher scores consistently indicate better performance.
} 
\label{fig:teaser}
\end{figure}

\clearpage
\tableofcontents
\newpage

\section{Introduction}
\label{sec:introd}

Rapid advances in large language model (LLM)~\cite{minaee2024large,yao2024survey,jiang2026survey} capabilities have given rise to an emerging paradigm of recursive self-improvement (RSI)~\cite{chen2026recursive,duan2026last,zheng2026dream}, in which models become the primary drivers of their own iterative improvement, enabling them to continuously adapt to increasingly complex scenarios and tasks~\cite{lu2024aiscientistfullyautomated,zhang2026darwin}. However, sustained growth in model capabilities does not automatically translate into corresponding gains in safety. Prior work has shown that even a single fine-tuning step can weaken a model's established alignment~\cite{qi2024fine,qi2025safety}; over the course of recursive evolution, repeated model updates may introduce new safety vulnerabilities and exacerbate existing safety issues~\cite{ibrahimouglu2026autonomous,gong2026evolutionary,zhang2026safety}. Nevertheless, current self-improvement methods predominantly optimize for benchmark performance, without systematically incorporating safety into the evolutionary loop, making it difficult to ensure that safety capabilities keep pace with ongoing changes in model capabilities. Leading AI organizations have begun to confront this capability–safety gap: the GPT-5.6~\cite{gpt56} and Claude 4~\cite{claude4} system cards assess not only model capabilities but also risks such as misuse, prompt injection, and alignment failures, while documenting the corresponding safeguards in safety training and deployment. These practices resonate with the central claim of our previously proposed AI-45° law~\cite{safeworkr1}: safety capabilities should co-evolve with model capabilities, rather than being applied as post hoc patches after model capabilities have stabilized.


Such coevolution requires safety alignment to continue throughout recursive model improvement—a process we term Safety RSI. As model capabilities and attack strategies evolve, static alignment against a fixed set of risks becomes insufficient~\cite{wu2025coga,zhou2025robust,wang2026lifelong,sunposition}. Safety RSI therefore requires continual discovery and repair of emerging vulnerabilities while preserving utility and limiting over-refusal. OpenAI's GPT-Red illustrates how self-play-trained red-teamers can generate attacks that improve subsequent models' prompt-injection robustness~\cite{gpt56,GPT-Red}. However, this continual discover-and-repair process is undermined by how current models encode and store knowledge: i) current architectures store safety knowledge in shared parameters rather than a dedicated substrate, so it is overwritten by subsequent updates as the model evolves and previously addressed risks are forgotten~\cite{qi2024fine,qi2025safety}; ii) current architectures encode hazardous and general knowledge in shared parameters, making it difficult to locate and isolate hazardous knowledge as models improve through RSI. Interventions on these parameters may impair general capabilities~\cite{wmdp,zhang2024negative}, while subsequent fine-tuning can restore previously suppressed hazardous knowledge~\cite{doshi2024does}. Addressing these limitations requires architectural support for safety memory and hazardous-knowledge disentanglement: the former retains safety experience for use in later updates, while the latter separates hazardous knowledge from general knowledge so it can be removed with minimal impact on general capabilities. 

These mechanisms support the continuous development and preservation of safety capabilities throughout Safety RSI. However, the resulting alignment remains primarily empirical: models learn safety responses from risk-related data, without necessarily establishing why a response satisfies the relevant safety requirements. Formal reasoning complements this empirical alignment by enabling the system to reason explicitly about safety-relevant conditions and derive conclusions from specified premises and rules. An inference-time neuro-symbolic harness can implement this capability by combining the model's language understanding with symbolic solvers that derive conclusions and construct checkable proofs under explicit assumptions~\cite{10.1145/1592434.1592436}. These proofs provide a machine-checkable basis for safety judgments beyond learned refusal behavior, with guarantees scoped to the formalized assumptions and safety requirements. The applicability of these guarantees to the original problems depends on how faithfully the formal representations capture their meaning. In existing neuro-symbolic approaches, translating natural-language problems into formal representations can omit implicit premises or semantic relations, while variations in wording can lead to inconsistent formalizations, preventing the solver from establishing the intended conclusion~\cite{pan-etal-2023-logic,jiang2023draft}. Addressing this gap must therefore preserve the connection between natural-language meaning and formal representations, making both the formalized premises and the semantic bridges used in reasoning explicit so that the resulting proofs can be checked and interpreted within their stated scope.

\begin{figure}[t]
\centering
\includegraphics[width=0.96\textwidth]{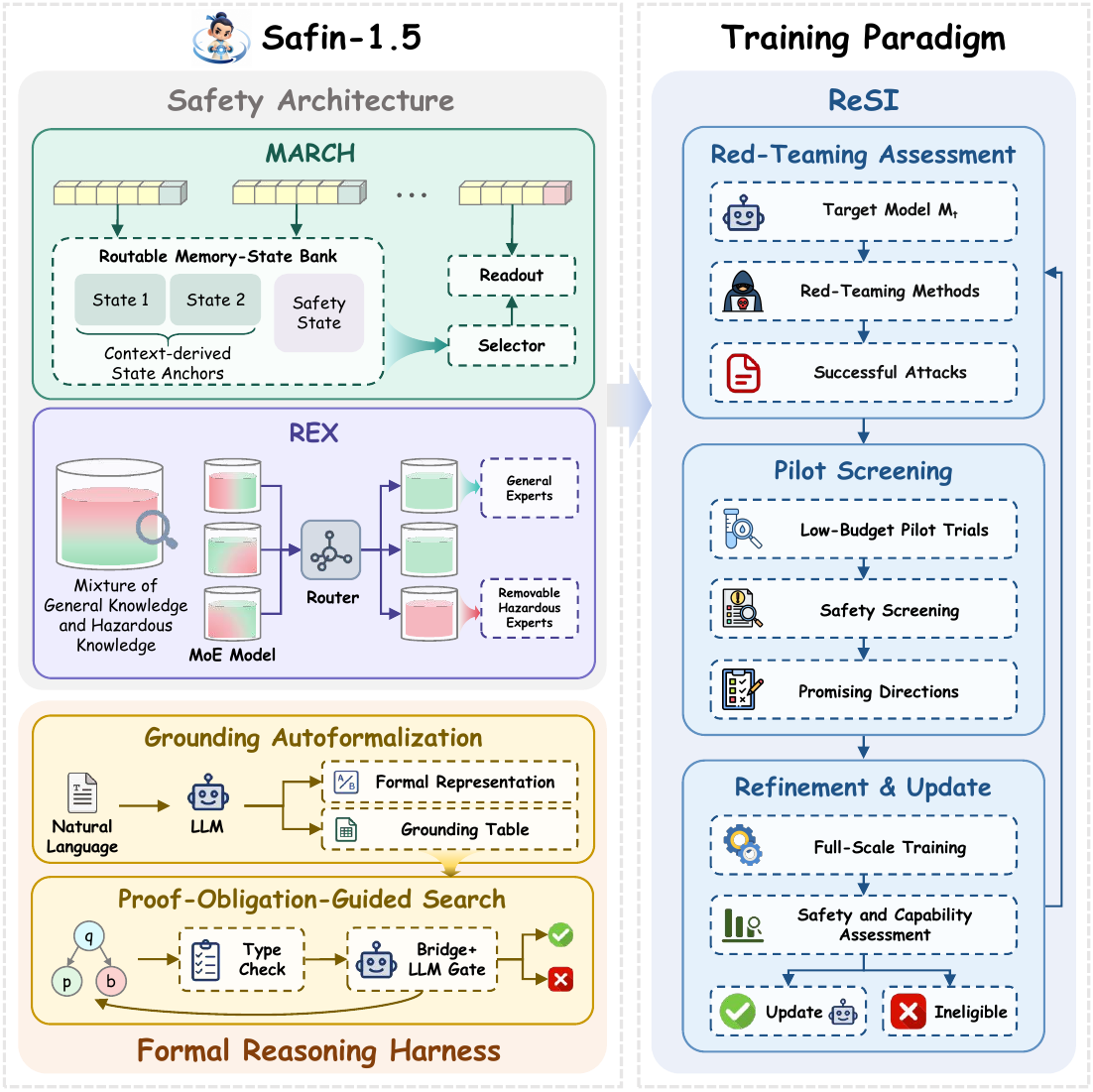}
\caption{Overview of \Projectname{}.
At the architectural level, \textbf{MARCH} and \textbf{REX} provide persistent safety memory and hazardous-knowledge disentanglement, respectively. At the training level, \textbf{ReSI} supports continual safety improvement through iterative red-teaming, pilot screening, and full-scale model updates. At the formal assurance level, the formal reasoning harness leverages \textbf{POGB} to address semantic gaps between natural-language problems and formal representations, enhancing the model's logical reasoning capabilities.
} 
\label{fig:framework}
\end{figure}

Building on these requirements, we introduce \Projectname{}, the latest generation of the Safin series. Its predecessor, Safin-1~\cite{safin1}, offered an initial exploration of Safety from Within, demonstrating the potential of integrating persistent safety capabilities into a model's internal architecture. \Projectname{}~extends this foundation into a full-stack framework for operationalizing Safety RSI by co-designing model architecture, continual training, and formal assurance, as illustrated in Figure~\ref{fig:framework}. At the architectural level, we integrate \textbf{Memory-Anchor Routing across Context History (MARCH)}~\cite{zhang2026march} and \textbf{Relocating into Removable EXperts (REX)}~\cite{jin2026architecturallevelsafetycontrolrelocating} to support persistent safety memory and hazardous-knowledge disentanglement, respectively. Together, these mechanisms enable safety experience to be retained and reused across model updates while facilitating the targeted removal of hazardous knowledge with minimal disruption to general capabilities. At the training level, \textbf{Recursive Safety Improvement (ReSI)}~\cite{resi} operationalizes Safety RSI through a continual cycle of risk discovery, evaluation, and repair, allowing safety alignment to adapt to emerging vulnerabilities while balancing safety improvements, capability retention, and over-refusal. The architectural mechanisms underpin this cycle by preserving prior safety experience and enabling targeted interventions against newly identified risks. At the formal assurance level, \textbf{Proof-Obligation-Guided Bridging (POGB)}~\cite{sun2026pogb} serves as the core mechanism of our Formal Reasoning Harness, complementing empirical safety evaluation with verifiable reasoning evidence at inference time. By addressing semantic gaps between natural-language problems and formal representations, it enables machine-checkable verification of conclusions in logical reasoning tasks under explicit assumptions.

Our main contributions are as follows:

\begin{itemize}
    \item We build on Safin-1's safety memory and incorporate REX-based hazardous-knowledge disentanglement to support continual safety improvement. 
    i) The inherited MARCH interface makes persistent safety states selectively accessible within model computation; ii) REX relocates the hazardous knowledge already present in a pretrained MoE model into newly added removable experts, keeps general capability in the original experts, trains subsequent updates on the full model, and physically excises the removable experts at deployment. On Safin-1, a single set of removable experts carries both Bio and Cyber knowledge; one excision substantially lowers WMDP accuracy in both domains, and the removal persists through standard post-training and \ReSI{}, whereas the same RMU recipe rebounds almost completely under the same post-training.
    \item We integrate \ReSI{}as the training harness for Safety RSI. In each round, \ReSI{} applies diverse red-teaming methods to identify vulnerabilities in the current target model, develops a training recipe through automated research, and promotes the update with the largest safety gain among those passing a Pareto gate on benign compliance and instruction following. On the disentangled architecture, \ReSI{} trains the full model while red-teaming, validating, and gating its excised counterpart, and its extensible pools of red-teaming and alignment methods let the loop incorporate new attacks and learning approaches.
    \item To support trustworthy safety improvement beyond empirical evaluation, we enhance the reliability of logical reasoning through machine-checkable proofs. We introduce POGB, which addresses semantic gaps between natural-language statements and formal representations by integrating LLM-based semantic understanding with solver-controlled logical reasoning. Extensive experiments demonstrate substantial improvements in reasoning accuracy across diverse benchmarks.
    \item Experimental results show that, compared with Safin-1, \Projectname{} reduces H-CoT ASR from 82.00\% to 14.00\% and X-Teaming ASR from 64.78\% to 32.08\%, lower than all evaluated frontier models on X-Teaming, while largely preserving general capabilities and benign compliance. At test time, POGB-based Formal Reasoning Harness achieves 91.80\% overall accuracy across four reasoning benchmarks, surpassing all evaluated Direct and CoT baseline configurations except those of Claude Opus 5. It achieves the highest FOLIO accuracy of 88.73\%, matches the best PrOntoQA result at 100.00\%, and reaches near-perfect accuracy of 99.50\% on ProofWriter, demonstrating the benefit of proof-guided reasoning for end-to-end accuracy.
\end{itemize}
 
\section{Safety-Native Architecture}
\label{sec:safety}
This section summarizes the safety memory inherited from Safin-1 and describes its integration with REX-based hazardous-knowledge disentanglement. Together, these mechanisms provide architectural support for retaining safety experience and isolating targeted hazardous knowledge during continual safety improvement. 

\subsection{MARCH: Safety Memory}
\label{sec:march}

\Projectname{} retains the memory architecture of MARCH~\cite{zhang2026march} and the persistent Safety State introduced in Safin-1~\cite{safin1}. This interface makes both contextual memory and learned safety specialization selectively accessible within recurrent computation. We describe one state-bearing layer below; the same construction applies across recurrent layers.

\paragraph{State anchors.}
Let $\mathbf{x}_t$ be the normalized token representation and $\mathbf{S}_t\in\mathbb{R}^{d_v\times d_k}$ the recurrent state, with native read query $\mathbf{q}_t=\mathbf{W}_q\mathbf{x}_t\in\mathbb{R}^{d_k}$. The base mixer updates $\mathbf{S}_t=F(\mathbf{S}_{t-1},\mathbf{x}_t)$ continuously, and MARCH checkpoints its state at boundaries $b_m$ spaced every $C$ text tokens:
\begin{equation}
\mathbf{A}^{(m)}=\mathbf{S}_{b_m},
\qquad \mathcal{V}_t=\{m\mid b_m<t\}.
\label{eq:march_anchors}
\end{equation}
Each anchor stores a cumulative prefix state rather than an isolated segment. An occurrence of a shared learned anchor embedding is inserted after each boundary without modifying the recurrence. At each layer, its normalized representation $\mathbf{u}_m$ reads the aligned checkpoint as $\mathbf{A}^{(m)}\mathbf{W}_q\mathbf{u}_m$ and yields a routing key $\boldsymbol{\kappa}_m=\mathbf{W}_{\kappa}\mathbf{u}_m\in\mathbb{R}^{d_r}$. The readout passes through the output projection and residual pathway, making subsequent-layer anchor representations state-dependent.

\paragraph{Persistent safety memory.}
The bank additionally contains a learned state $\mathbf{P}_{\mathrm{safe}}\in\mathbb{R}^{d_v\times d_k}$, indexed by $p_{\mathrm{safe}}$, with a key constructed through the same anchor-representation pathway. The layer-wise collection of these states forms the Safety State. Unlike sequence-derived anchors, it remains available from the first token without consuming input positions. Safin-1 learns this collection with token-level cross-entropy on harmful requests with safety-aligned responses and benign requests with helpful responses, while freezing the language-model backbone~\cite{safin1}.

\paragraph{Joint routing and readout.}
Let $\boldsymbol{\rho}_t=\mathbf{W}_{\rho}\mathbf{x}_t\in\mathbb{R}^{d_r}$ be the routing query. The candidate set is $\mathcal{C}_t=\mathcal{V}_t\cup\{p_{\mathrm{safe}},\varnothing\}$, with payloads $\mathbf{A}^{(p_{\mathrm{safe}})}=\mathbf{P}_{\mathrm{safe}}$ and $\mathbf{A}^{(\varnothing)}=\mathbf{0}$. The null option allows retrieval to be bypassed. Candidate scores and probabilities are
\begin{equation}
\begin{aligned}
s_{t,j}&=\begin{cases}
\boldsymbol{\rho}_t^{\top}\boldsymbol{\kappa}_j,&j\ne\varnothing,\\
\mathbf{w}_{\varnothing}^{\top}\mathbf{x}_t+b_{\varnothing},&j=\varnothing,
\end{cases}
&\pi_{t,j}&=\frac{\exp(s_{t,j})}{\sum_{j'\in\mathcal{C}_t}\exp(s_{t,j'})}.
\end{aligned}
\label{eq:march_routing}
\end{equation}
The router selects memories, while $\mathbf{q}_t$ reads their matrix-valued contents. The retrieved readout is added to the native recurrent path:
\begin{equation}
\mathbf{o}_t=\mathbf{S}_t\mathbf{q}_t+
\sum_{j\in\mathcal{C}_t}\pi_{t,j}\mathbf{A}^{(j)}\mathbf{q}_t.
\label{eq:march_readout}
\end{equation}
Thus, context anchors and persistent safety memory contribute according to the current token and context, without changing the base recurrent update. In \Projectname{}, this inherited interface supplies the memory foundation alongside REX and \ReSI{}; replay and the safety--utility gate in Section~\ref{sec:resi_arch} complement it during continual safety improvement.

\subsection{REX: Hazardous-Knowledge Disentanglement}
\label{sec:rex}

\begin{figure}[t]
\centering
\includegraphics[width=\textwidth]{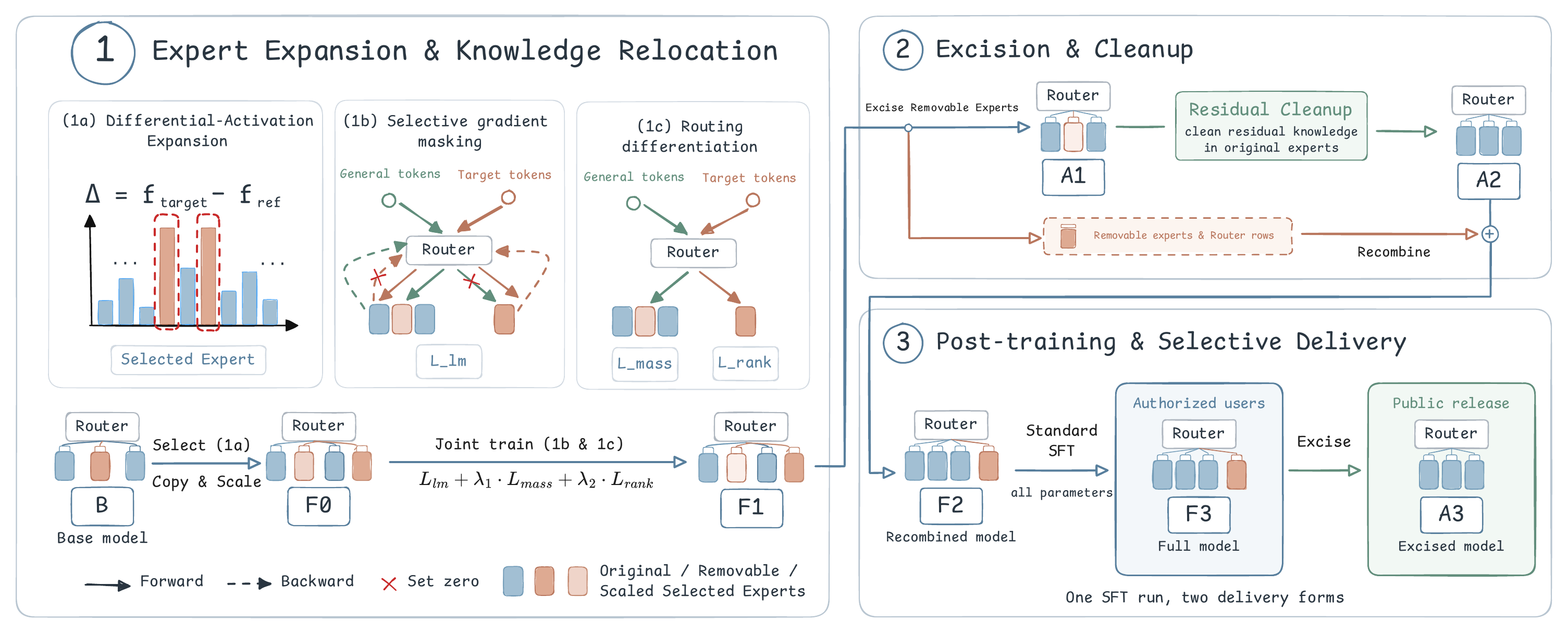}
\caption{\textbf{Overview of REX.}
(1) Expert expansion copies and scales the source experts of the base model $B$ to form removable experts ($F_0$), and joint training with selective gradient masking and routing differentiation relocates hazardous knowledge into them ($F_1$).
(2) Excising the removable experts yields $A_1$, and a light RMU cleanup yields the public base model $A_2$, while the removable experts and their router rows are kept.
(3) Recombining them with $A_2$ gives the full structure $F_2$, on which standard post-training produces the full model $F_3$ for authorized users; excising its removable experts produces the public model $A_3$.}
\label{fig:rex}
\end{figure}

Safety RSI requires that hazardous capabilities, once removed, stay removed as the model keeps being updated, and parameter-level unlearning has difficulty meeting this requirement because hazardous and general knowledge are stored in the same parameters.
Methods such as RMU~\cite{wmdp} and NPO~\cite{zhang2024negative} modify these shared parameters in place, so more thorough removal costs more general capability.
More importantly, such removal suppresses the expression of the knowledge rather than taking it out of the model, and subsequent fine-tuning, even on unrelated data, can bring it back~\cite{doshi2024does}.
In a one-off release this is already a risk; in Safety RSI, where the model is updated round after round, every update is another chance for the knowledge to return.
REX~\cite{jin2026architecturallevelsafetycontrolrelocating} addresses this problem through where the knowledge is stored.
The feed-forward parameters of a Mixture-of-Experts (MoE) model are partitioned into independently addressable experts, the expert set can be expanded, and selected experts can be deleted at deployment.
Once hazardous knowledge is concentrated in a set of removable experts, removing it no longer depends on in-place edits to shared parameters and becomes a structural operation.

REX relocates hazardous knowledge out of the original experts and into newly added removable experts through three components (Figure~\ref{fig:rex}).
The first component, \emph{expert expansion}, initializes the removable experts.
On the training corpora, we count how often each expert $i$ is selected into the top-$K$ on target data $\mathcal{D}_{\mathrm{harm}}$ and on general data {$\mathcal{D}_{\mathrm{gen}}$}, and compute the difference $\Delta_i = f_i^{\mathrm{harm}} - {f_i^{\mathrm{gen}}}$.
In each layer, the experts with the largest $\Delta_i$ serve as source experts.
Their parameters and router rows are copied to form the removable experts, and the source experts are scaled by a fixed factor {$\gamma$}, which weakens the hazardous knowledge they hold while preserving most of their general computation.
The second component, \emph{selective gradient masking}, determines where the updates from each data source can be written.
Let {$\Theta_{\mathrm{rem}}$} denote the removable experts, {$\Theta_{\mathrm{orig}}$} the original experts, and {$\Theta_{\mathrm{router}}$} all router rows; the updated parameters are
\begin{equation}
\begin{cases}
{\Theta_{\mathrm{rem}}} \cup {\Theta_{\mathrm{router}}} & \text{on } \mathcal{D}_{\mathrm{harm}},\\
{\Theta_{\mathrm{orig}}} \cup {\Theta_{\mathrm{router}}} & \text{on } {\mathcal{D}_{\mathrm{gen}}}.
\end{cases}
\label{eq:rex_mask}
\end{equation}
On general data, a pre-top-$K$ activation mask additionally removes the logits of the removable experts before selection, so that every general token is dispatched without them.
The original experts must therefore carry general capability on their own, and the hazardous knowledge they already hold receives no training signal and fades, which turns catastrophic forgetting during continued training into directed forgetting.
The third component, \emph{routing differentiation}, shapes the router with two losses.
A mass loss keeps the routing mass assigned to the removable experts above a floor on target tokens and below a ceiling on general tokens.
A rank loss keeps the largest removable-expert logit on general tokens a margin below the top-$K$ boundary of the original experts, so that the unmasked router at inference also leaves general tokens to the original experts.
The relocation objective is
\begin{equation}
\mathcal{L}_{\mathrm{REX}} = \mathcal{L}_{\mathrm{LM}} + \lambda_1 \mathcal{L}_{\mathrm{mass}} + \lambda_2 \mathcal{L}_{\mathrm{rank}},
\label{eq:rex_loss}
\end{equation}
where $\lambda_1$ and $\lambda_2$ weight the mass and rank losses, respectively, and the language-modeling loss is back-propagated under the masking rules of Eq.~\eqref{eq:rex_mask}.
During relocation, attention, embeddings, normalization layers, shared experts, and the LM head stay frozen, and only the routed experts and the router are trained.

After relocation, REX excises the removable experts together with their router rows and applies a light RMU cleanup to the retained parameters to remove residual target knowledge, which yields the public base model; the excised experts are kept by the developer.
The value of this design for Safety RSI lies in the updates that follow.
An update is not applied to the public model directly.
Instead, the kept removable experts are recombined with the cleaned parameters into the full structure, training proceeds on the full structure, and the removable experts are excised again at deployment.
During the update, the router continues to send target-related tokens to the removable experts, so these updates leave with the experts at deployment, while growth in general capability is written to the retained parameters and released with the public model.
The developer therefore always holds two forms of the model: a full model for authorized users that retains the hazardous capability, and a public model that contains no parameters of the removable experts.
What to train and what to release thus become two independent decisions, which is exactly the interface that the \ReSI{} loop needs when it accepts an update in each round (Section~\ref{sec:resi_arch}).
In the experiments, we denote the original model by $B$, the cleaned public base model by $A_2$, the recombined full model after standard post-training by $F_3$, and its excised counterpart by $A_3$.

\Projectname{} applies REX to Safin-1-35B-A3B without interfering with the memory mechanism of MARCH.
The memory-related parameters introduced by MARCH belong to the frozen shared parameters during both relocation and cleanup, so REX changes only the routed experts and the router.
Routing statistics, training-token counts, and RMU activations are computed on text positions, excluding the anchor positions that MARCH inserts into the sequence.
Unlike the original REX, which handles one domain at a time, \Projectname{} uses a single set of removable experts for both Bio and Cyber.
Expert expansion selects 8 source experts per layer for each domain by differential activation, takes their union, and completes it to 16 per layer by the joint differential across the two domains, expanding each layer from 256 to 272 experts with $\gamma=0.3$.
Relocation uses about 40M tokens of mixed target data, with Bio and Cyber at roughly 4:1, and about 27M tokens of general data.
The residual cleanup uses a forget budget of 32M tokens and updates only the routed-expert and shared-expert MLPs in layers 7--9, and standard post-training on the recombined model uses 250M tokens.
A single excision then removes the target knowledge of both domains.
 
\section{ReSI: Recursive Safety Improvement}
\label{sec:resi}

\ReSI{}~\cite{resi} builds on the safety architecture described in Section~\ref{sec:safety} to implement the risk discovery, evaluation, and repair loop of Safety RSI.
In each round, \ReSI{} applies diverse red-teaming methods to identify vulnerabilities in the current target model, develops a training recipe through automated research, and promotes the eligible update with the largest safety gain as the next target model, where eligibility requires passing a Pareto gate on safety gain, benign compliance, and instruction following (Figure~\ref{fig:resi}).
Through successive rounds, the loop pursues the two goals identified by R$^2$AI~\cite{r2ai2025,sunposition}: \emph{resistance} to known threats and \emph{resilience} to unforeseen risks.
Section~\ref{sec:resi_formulation} formulates the recursive update, Section~\ref{sec:resi_round} describes the three stages of a round and the joint recipe search space, and Section~\ref{sec:resi_arch} describes how the loop operates on the safety architecture.

\begin{figure}[htbp]
\centering
\includegraphics[width=\linewidth]{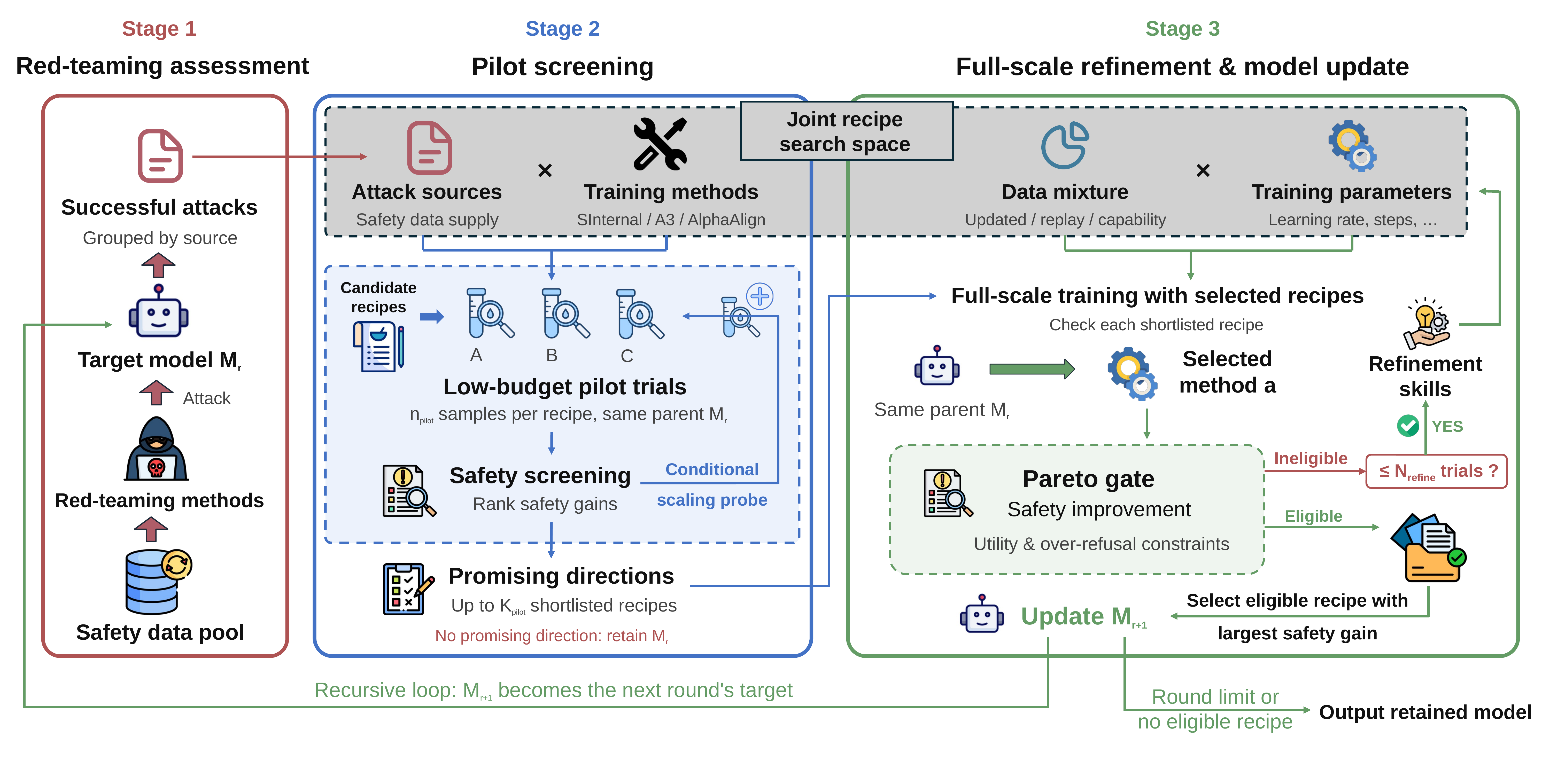}
\caption{\textbf{The three-stage \ReSI{} framework.}
In each round, \ReSI{} applies red-teaming methods to identify vulnerabilities in the current target model $M_r$, develops training recipes through automated research, and promotes the eligible update with the largest safety gain as $M_{r+1}$, the target for the next round.
Stage~1 collects successful attacks against $M_r$ (Section~\ref{sec:resi_redteam}).
Stage~2 screens attack sources and alignment methods through low-budget pilot trials and shortlists the most promising recipes (Section~\ref{sec:resi_pilot}).
Stage~3 conducts full-scale training and feedback-guided refinement, then selects the eligible checkpoint with the largest safety gain (Section~\ref{sec:resi_full}).}
\label{fig:resi}
\end{figure}

\subsection{Recursive Formulation}
\label{sec:resi_formulation}

\ReSI{} organizes safety improvement as a recursive loop over target models.
Specifically, \ReSI{} formulates the safety update within a round as an auto-research problem: given $M_r$ and the attacks that succeed against it, \ReSI{} searches for a training recipe that strengthens the model's resistance to the exposed attacks while preserving existing safety behavior and general capabilities.
Starting from an initial model $M_0$, round $r$ takes the current target model $M_r$ and produces the next target
\begin{equation}
M_{r+1}=
\begin{cases}
M_r^{\omega_r^\star}, & \mathcal{E}_r\neq\varnothing,\\[2pt]
M_r, & \mathcal{E}_r=\varnothing,
\end{cases}
\qquad
\omega_r^\star\in\operatorname*{arg\,max}_{\omega\in\mathcal{E}_r} g_r(\omega),
\label{eq:resi_update}
\end{equation}
where $M_r^{\omega}$ denotes the checkpoint that recipe $\omega$ produces in round $r$ from its specified starting checkpoint, $g_r(\omega)$ is its safety gain (Eq.~\eqref{eq:resi_gain}), and $\mathcal{E}_r$ is the set of recipes with completed full-scale trials, including the refinement trials of Section~\ref{sec:resi_full}, that pass the Pareto gate (Eq.~\eqref{eq:resi_gate}).
As a result, an accepted update becomes the target of the next round; if no recipe passes the gate, \ReSI{} retains $M_r$ and terminates, so that updates failing the gate never propagate to later rounds.
We use \emph{recursive} in a bounded, operational sense: the search procedure and the Pareto gate remain fixed across rounds, while the target model and the attacks collected against it change.

\subsection{The Three-Stage Round}
\label{sec:resi_round}

Each round of \ReSI{} proceeds in three stages over a joint recipe search space (Figure~\ref{fig:resi}).
Stage~1 collects successful attacks against $M_r$ (Section~\ref{sec:resi_redteam}), which define the attack data of the recipe search space (Section~\ref{sec:resi_space}).
Within this space, Stage~2 screens attack sources and alignment methods through pilot trials (Section~\ref{sec:resi_pilot}), and Stage~3 refines full-scale recipes and selects the model update under the Pareto gate (Section~\ref{sec:resi_full}).

\subsubsection{Stage 1: Red-Teaming Assessment}
\label{sec:resi_redteam}

Stage~1 identifies the vulnerabilities of the current target model and converts them into safety training signals.
Let $\mathcal{X}_r$ denote a safety data pool of requests that represent potential deployment risks, and let $\mathcal{Z}_r$ denote the set of available red-teaming methods.
Each method $z\in\mathcal{Z}_r$ draws requests from $\mathcal{X}_r$ to construct attacks against $M_r$, potentially adapting them based on the model's responses, and produces input--response pairs $z(M_r,\mathcal{X}_r)$.
\ReSI{} aggregates these records and retains those in which the attack elicits an unsafe response:
\begin{equation}
\mathcal{D}_r^{\mathrm{atk}}=
\Bigl\{(x,y)\in\textstyle\bigcup_{z\in\mathcal{Z}_r}z(M_r,\mathcal{X}_r)\;\Big|\;\operatorname{Unsafe}(x,y)=1\Bigr\},
\end{equation}
where $x$ is an attack input, $y$ is the corresponding response from $M_r$, and $\operatorname{Unsafe}(x,y)$ indicates a safety violation; $\mathcal{D}_{r,z}^{\mathrm{atk}}$ denotes the subset contributed by method $z$.
Since the attack dataset is re-collected for each target model, it captures the vulnerabilities of the current checkpoint, including those introduced by earlier updates.
Moreover, the safety data pool and the set of red-teaming methods are both extensible, so that new deployment risks and new red-teaming methods can enter the loop as they emerge.

\subsubsection{Joint Recipe Search Space}
\label{sec:resi_space}

Following Stage~1, \ReSI{} searches over a joint recipe space $\Omega_r$ of compatible combinations of attack data, alignment methods, training data mixtures, and hyperparameters.
A candidate recipe is represented as
\begin{equation}
\omega=(\boldsymbol{\alpha},a,\boldsymbol{\lambda},\mathbf{h})\in\Omega_r,
\end{equation}
where $\boldsymbol{\alpha}$, $a$, $\boldsymbol{\lambda}$, and $\mathbf{h}$ specify the attack data, alignment method, training data mixture, and hyperparameters, respectively.

\textbf{Attack data.}
The attack-data dimension $\boldsymbol{\alpha}$ controls the contribution of each red-teaming source.
For each method $z$, $\alpha_z\in[0,1]$ selects a fraction of $\mathcal{D}_{r,z}^{\mathrm{atk}}$, so that a recipe can use a single source, combine all available attacks, or reweight the sources.

\textbf{Alignment methods.}
The alignment method $a$ is drawn from an extensible pool $\mathcal{M}_{\mathrm{align}}$ and converts the selected attacks into training examples.
Representative methods in this pool and their design rationales are outlined below.
A3~\cite{a3_2026} pairs harmful requests with benign counterparts that share similar patterns and uses both for supervised fine-tuning (SFT), encouraging the model to learn safety distinctions that generalize beyond specific adversarial patterns.
SInternal~\cite{sinternal2026} fine-tunes the model to reproduce an expert's critique of its own response against safety specifications and then predict a safety judgment, strengthening its understanding of the specifications.
AlphaAlign~\cite{alphaalign2025} applies reinforcement learning with a reward that requires safe responses to harmful requests and safe, fully compliant responses to benign requests of similar surface form, so that unnecessary refusals receive no reward.

\textbf{Training data mixtures.}
The mixture $\boldsymbol{\lambda}=(\lambda_{\mathrm{cur}},\lambda_{\mathrm{rep}},\lambda_{\mathrm{gen}})$ balances learning from newly identified vulnerabilities with retaining existing safety behavior and general capabilities.
Specifically, the three weights of $\boldsymbol{\lambda}$ are non-negative, sum to one, and give the fractions of current safety examples, replay examples from earlier safety training, and general-task examples.

\textbf{Hyperparameters.}
The configuration $\mathbf{h}\in\mathcal{H}(a)$, drawn from the training settings available for method $a$, covers the learning rate, batch size, training budget, checkpoint frequency, the evaluated checkpoint step, adaptation settings such as LoRA~\cite{lora2021}, and generation parameters for reinforcement learning.

\subsubsection{Stage 2: Pilot Screening}
\label{sec:resi_pilot}

Pilot screening identifies promising attack sources and alignment methods through low-budget trials.
For each source $z$ providing at least $n_{\mathrm{pilot}}$ successful attacks and each method $a\in\mathcal{M}_{\mathrm{align}}$, \ReSI{} trains a pilot from $M_r$ with method $a$ on $n_{\mathrm{pilot}}$ attacks from $z$, using only current safety data ($\boldsymbol{\lambda}^{(0)}=(1,0,0)$), the default hyperparameters $\mathbf{h}_a^{(0)}$ of $a$, and a limited training budget.
In the first round, successful attacks from the original safety data pool form a shared base that each pilot uses in addition to its $n_{\mathrm{pilot}}$ source-specific attacks.
After training, each recipe is evaluated on a held-out validation set $\mathcal{D}_r^{\mathrm{val}}$ drawn from the original safety data pool.
However, when $M_r$ already has a low attack success rate (ASR) on that pool, $\mathcal{D}_r^{\mathrm{val}}$ is instead drawn from a more challenging attack set, and a candidate must also not raise ASR on the original pool by more than a tolerance $\delta$ relative to $M_r$.
The safety gain of $\omega$ is
\begin{equation}
g_r(\omega)=\operatorname{ASR}(M_r;\mathcal{D}_r^{\mathrm{val}})-\operatorname{ASR}(M_r^{\omega};\mathcal{D}_r^{\mathrm{val}}).
\label{eq:resi_gain}
\end{equation}
Although equal sample counts keep pilots comparable, they can understate a source that provides many more attacks than $n_{\mathrm{pilot}}$.
Therefore, for recipes with modest pilot safety gains but substantially more available attacks, \ReSI{} runs a conditional scaling probe, retraining the same recipe from $M_r$ with $n_{\mathrm{pilot}}<n_{\mathrm{scale}}\leq|\mathcal{D}_{r,z}^{\mathrm{atk}}|$ attacks and ranking it alongside the initial pilots.
Finally, \ReSI{} shortlists up to $K_{\mathrm{pilot}}$ recipes with the largest positive safety gains, $\mathcal{Q}_r=\operatorname{TopK}_{g_r}\{\omega:g_r(\omega)>0\}$; if $\mathcal{Q}_r=\varnothing$, it retains $M_r$ and terminates.

\subsubsection{Stage 3: Full-Scale Refinement and Model Update}
\label{sec:resi_full}

Given the shortlisted recipes $\mathcal{Q}_r$, \ReSI{} conducts full-scale trials, refines each recipe based on feedback from the Pareto gate, and selects the model update.
Specifically, each full-scale trial starts from $M_r$, expands the training data to all successful attacks from the source of the shortlisted recipe, and lifts the pilot budget cap, keeping $\boldsymbol{\lambda}^{(0)}$ and $\mathbf{h}_a^{(0)}$.
The trial saves checkpoints at the frequency specified in $\mathbf{h}$, and each saved checkpoint is evaluated as a candidate update.

\textbf{Pareto gate.}
The Pareto gate requires a safety improvement over the current model while keeping benign compliance and instruction following within a tolerance of the initial model.
A recipe passes the gate if
\begin{equation}
\mathcal{G}_r(\omega)=1
\quad\Longleftrightarrow\quad
\begin{cases}
g_r(\omega)>0,\\
\mathrm{BC}(M_r^{\omega})\geq\mathrm{BC}(M_0)-\epsilon,\\
\mathrm{IF}(M_r^{\omega})\geq\mathrm{IF}(M_0)-\epsilon,
\end{cases}
\label{eq:resi_gate}
\end{equation}
where $\mathrm{BC}$ denotes benign compliance, $\mathrm{IF}$ denotes instruction following, and $\epsilon$ is the retention tolerance.
Here, the benign-compliance condition limits over-refusal, and the instruction-following condition limits the loss of general capabilities.
Moreover, anchoring both retention thresholds to $M_0$ rather than to $M_r$ prevents small per-round losses within the tolerance from accumulating across rounds.

\textbf{Feedback-guided refinement.}
\ReSI{} revises a recipe according to its Pareto-gate result, allowing at most $N_{\mathrm{refine}}$ refinement trials beyond the initial trial and stopping earlier when no executable revision remains.
\begin{itemize}[leftmargin=*,itemsep=1pt,topsep=2pt]
\item \textbf{S1. Attack-source composition.} For limited safety gains, expand $\boldsymbol{\alpha}$ to combine successful attacks from additional promising sources.
\item \textbf{S2. Harmful--benign pairing.} For a benign-compliance violation, construct A3-style benign counterparts that preserve the attack patterns of harmful examples while replacing their harmful content, and adjust the harmful--benign balance.
\item \textbf{S3. Utility preservation.} For an instruction-following violation, increase general-task replay through $\boldsymbol{\lambda}$, or reduce the learning rate or the number of updates through $\mathbf{h}$.
\item \textbf{S4. Post-SFT safety GRPO.} For an eligible SFT recipe, retain it as a candidate and continue training from its checkpoint with AlphaAlign-based safety Group Relative Policy Optimization (GRPO)~\cite{deepseekmath2024} and DAPO-style optimization~\cite{yu2025dapo}; this transition resets the refinement counter to zero, allowing up to $N_{\mathrm{refine}}$ further refinement trials.
\end{itemize}
In practice, each revision specifies its starting checkpoint, while all trials share the same validation sets and gate thresholds.
Revisions restart from $M_r$ by default; a revision may instead continue from the evaluated checkpoint with a reduced training budget recorded in $\mathbf{h}$, and revisions within an S4 continuation restart from the same SFT checkpoint.

\textbf{Final selection and stopping.}
After all shortlisted recipes and their refinements have been evaluated, \ReSI{} updates the target model according to Eq.~\eqref{eq:resi_update}.
\ReSI{} promotes the eligible checkpoint with the largest safety gain, or retains $M_r$ and terminates when no recipe is eligible; the loop also terminates at the configured round limit.

\subsection{\ReSI{} on the Safety Architecture}
\label{sec:resi_arch}

\ReSI{} operates on the safety architecture of Section~\ref{sec:safety} without changing the three stages of a round or the Pareto gate: hazardous-knowledge disentanglement determines which model is trained and which is evaluated, safety memory supports the retention and reuse of safety experience across rounds, and formal assurance can check specified claims about the models that the loop promotes.

\textbf{Hazardous-knowledge disentanglement.}
On the disentangled architecture, \ReSI{} trains the full model but red-teams, validates, and gates its excised counterpart; the full checkpoint of an accepted update becomes the starting point of the next round, and its excised counterpart is the deployed model.
Specifically, REX~\cite{jin2026architecturallevelsafetycontrolrelocating} relocates targeted hazardous knowledge into removable experts that are excised for deployment.
Therefore, in each round, a recipe updates the full model with all experts, so that the router and the removable experts keep the roles assigned during relocation.
In contrast, training without the removable experts could push hazardous content contained in safety data into shared parameters, whereas evaluating with them would assess a model that is never deployed.
Additionally, \ReSI{} measures hazardous-knowledge accuracy for each trial on the excised checkpoint, so that any reintroduction of removed knowledge is visible when candidates are compared.

\textbf{Safety memory.}
Safety memory keeps safety experience from earlier rounds available to later updates.
On the model side, MARCH~\cite{zhang2026march} can provide a persistent interface through which this experience is stored and selectively invoked.
On the training side, the replay weight $\lambda_{\mathrm{rep}}$ complements this interface by keeping examples from earlier safety training in the training data to reduce forgetting of earlier repairs.

\textbf{Formal assurance.}
Formal assurance can complement the empirical evaluation of each model that \ReSI{} promotes.
\ReSI{} provides empirical evidence that an update reduces ASRs on the evaluated attacks, and specified claims about each promoted model can then be checked by the formal assurance harness in Section~\ref{sec:formal} under explicit assumptions.
 
\section{POGB: Formal Reasoning Harness}
\label{sec:formal}
Building on the model trained through the safety improvement process in Section~\ref{sec:resi}, we introduce POGB as an inference-time harness to extend its capabilities to formal reasoning, as illustrated in Figure~\ref{fig:pogb}.
An inference-time switch controls whether POGB is invoked; it is enabled only for logical reasoning tasks that ask whether a claim follows from given premises. When disabled, the model follows its standard inference procedure. 
When enabled, PoGB begins with a grounding autoformalization step, in which an LLM translates the problem expressed in NL into formal representations while maintaining a grounding table that records the type and associated NL information for each formal element.
This table then serves as a \emph{proof-time interface} for \emph{proof-obligation-guided search}.
During this search, type information constrains which instantiations are admissible, while NL information guides the identification of semantic bridges via an LLM-based $\operatorname{Gate}$ when symbolic reasoning encounters an unresolved proof obligation.
This design keeps all logical reasoning under the solver's control while restricting the LLM's role to autoformalization and the evaluation of candidate semantic bridges.

\begin{figure}[!t]
    \centering
    \includegraphics[width=0.99\linewidth, trim=0cm 0.2cm 0cm 0.2cm, clip]{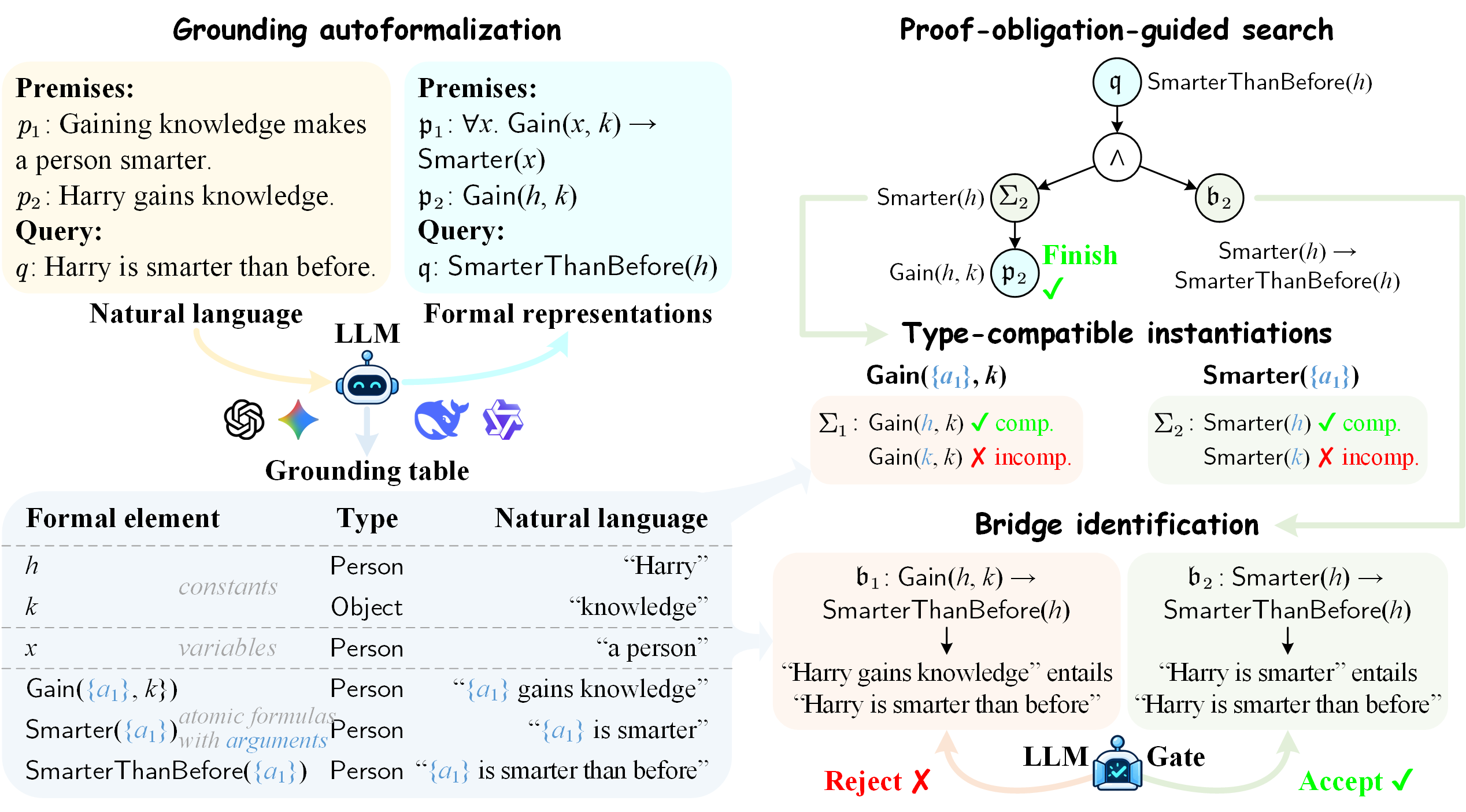}
    \caption{Illustration of the proposed approach.
    Variables in the atomic formulas of the grounding table are replaced with argument placeholders $\{a_i\}$, which are highlighted in blue and subsequently instantiated with type-compatible constants.}
    \label{fig:pogb}
\end{figure}

\subsection{Problem Formulation}
\label{sec:problem-formulation}

Let $\symbqa = (\symbPre,\symbhyp)$ denote a problem comprising a set of NL premises $\symbPre = \{\symbpre_i\}_{i=1}^n$ and an NL query $\symbhyp$.
We use a grounding autoformalization ${\symbAutofm}$ to map $\symbqa$ to ${\symbAutofm}(\symbqa) = (\symbPF,\symbHF,\symbSymTab)$, thereby yielding formal representations that are explicitly associated with their NL semantics in $\symbqa$.
Here, $\symbPF = \{\symbpf_i\}_{i=1}^n$ is the set of formalized premises, where $\symbpf_i$ denotes the $i$th formalized premise; $\symbHF$ is the formalized query; and $\symbSymTab$ is a grounding table that records the NL semantic grounding of every formal element (see Section~\ref{sec:grounding-interface} for details).
We adopt first-order logic (FOL) as the target formal language.

Consider two proof goals corresponding to the query and its negation: $\symbHF$ and $\neg\symbHF$.
For each goal $\symbgoal\in\{\symbHF,\neg\symbHF\}$, the symbolic solver performs a proof search and identifies unresolved proof obligations resulting from missing semantic relations.
Let $\symbSolver$ denote the symbolic solver.
Our objective is to find a \emph{reasonable} bridge set $\symbBridge$ that enables the proof to be completed:
\begin{equation}
\label{eq:accepted-completion}
\symbSolver(\symbPF\cup\symbBridge, \symbgoal) = \textsc{Proved}
\quad\land\quad
\operatorname{Gate}(\symbBridge, \symbSymTab) = \textsc{Accepted}.
\end{equation}
Here, $\symbBridge=\{\symbbridge_j\}_{j=1}^{m}$, where each $\symbbridge_j$ is a semantic bridge that explicitly encodes an entailment relation supported by the NL semantics but not represented in the symbolic ontology.
The LLM-based $\operatorname{Gate}$ evaluates only whether these candidate semantic bridges are reasonable given the grounding table $\symbSymTab$; it does not perform proof search, conduct intermediate reasoning, or repair failures in response to solver feedback.
The final decision is determined by the proof outcomes for the two goals:
\begin{equation*}
\operatorname{Decide}(\symbqa) =
\begin{cases}
\textsc{Entailed},
    & \bigl(\symbSolver(\symbPF\cup\symbBridge, \symbHF),\ \symbSolver(\symbPF\cup\symbBridge, \neg\symbHF)\bigr)
    = \bigl(\textsc{Proved},\ \textsc{Failed}\bigr),\\
\textsc{Contradicted},
    & \bigl(\symbSolver(\symbPF\cup\symbBridge, \symbHF),\ \symbSolver(\symbPF\cup\symbBridge, \neg\symbHF)\bigr)
    = \bigl(\textsc{Failed},\ \textsc{Proved}\bigr),\\
\textsc{Unknown},
    & \bigl(\symbSolver(\symbPF\cup\symbBridge, \symbHF),\ \symbSolver(\symbPF\cup\symbBridge, \neg\symbHF)\bigr)
    = \bigl(\textsc{Failed},\ \textsc{Failed}\bigr),\\
\textsc{Conflicted},
    & \bigl(\symbSolver(\symbPF\cup\symbBridge, \symbHF),\ \symbSolver(\symbPF\cup\symbBridge, \neg\symbHF)\bigr)
    = \bigl(\textsc{Proved},\ \textsc{Proved}\bigr).
\end{cases}
\label{eq:solver-decision}
\end{equation*}

\subsection{The Grounding Table as a Proof-Time Interface}
\label{sec:grounding-interface}

As discussed in Section~\ref{sec:introd}, even when individual NL statements are faithfully formalized, their semantically related expressions may still be mapped to symbolically unrelated predicates.
To recover the semantic relations between such predicates during proof search, we maintain two complementary kinds of information for every formal element.
Let $\symbterm$ denote a formal element occurring in $\symbPF$ or $\symbHF$, including a constant, a variable, or an atomic formula.
For each $\symbterm$, the grounding table $\symbSymTab$ records both type information, $\operatorname{type}_{\symbSymTab}(\symbterm)$, and NL information, $\operatorname{NL}_{\symbSymTab}(\symbterm)$, thereby \emph{providing an interface for proof-time inference}.

Inspired by type systems~\citep{cardelli1996type}, $\operatorname{type}_{\symbSymTab}(\symbterm)$ retains a coarse-grained semantic category assigned to $\symbterm$ by the LLM.
This category constrains which instantiations are admissible, thereby narrowing the search space for candidate semantic bridges.
For instance, in the running example in Appendix~\ref{app:exam}, a variable of type $\mathsf{Person}$ can be instantiated with $h$ (Harry) but not with $w$ (\emph{Walden}).
By contrast, $\operatorname{NL}_{\symbSymTab}(\symbterm)$ retains the finer-grained NL semantics of $\symbterm$ by storing referring expressions (e.g., “Harry,” “\emph{Walden},” “a person,” and “a book”) and predicate meanings (e.g., “$\{a_1\}$ reads $\{a_2\}$”).
When exact symbolic matching fails, this information guides the search for plausible semantic bridges between \emph{type-compatible} candidates, such as the bridge between $\mathsf{Smarter}(h)$ and $\mathsf{SmarterThanBefore}(h)$.
We next show how this interface supports \emph{proof-obligation-guided search}.

\subsection{Proof-Obligation-Guided Search}
\label{sec:proof-completion}

We propose a proof-search method controlled by a symbolic solver and guided by proof obligations, as described in Algorithms~\ref{alg:direct-search} and \ref{alg:grounding-retrieval}.
Algorithm~\ref{alg:direct-search} recursively decomposes each obligation through symbolic reasoning.
When symbolic reasoning alone is insufficient to discharge an obligation, Algorithm~\ref{alg:direct-search} invokes Algorithm~\ref{alg:grounding-retrieval} to identify local semantic bridges.
Algorithm~\ref{alg:grounding-retrieval} then returns the identified bridges together with the residual obligations required to establish each bridge, after which Algorithm~\ref{alg:direct-search} recursively discharges these residual obligations.
A running example is presented in Appendix~\ref{app:exam}.

\begin{algorithm}[!t]
\caption{Proof-obligation-guided search}
\label{alg:direct-search}
\begin{algorithmic}[1]\small
\Require A set of formalized premises~$\symbPF$, a symbolic solver~$\symbSolver$, a proof obligation~$\phi$, a set of ancestor obligations~$\symbAncestor$, and a grounding table~$\symbSymTab$.
\Ensure A set of accepted bridges sufficient to prove $\phi$ from $\symbPF$, or \textsc{Failed} if the search fails.
\Function{Prove}{$\phi, \symbAncestor, \symbSymTab$}
    \If{$\symbSolver(\symbPF, \phi) = \textsc{Proved}$}
    \Comment{Discharge $\phi$ directly with the symbolic solver}
        \State \Return $\varnothing$
    \EndIf
    \State $\phi \gets \Call{Normalize}{\phi}$
    \Comment{Rewrite $\phi$ in canonical form}
    \If{$\Call{Cyclic}{\phi, \symbAncestor}$}
    \Comment{Reject circular reasoning}
        \State \Return \textsc{Failed}
    \EndIf
    \State $\symbAncestor^\prime \gets \symbAncestor \cup \{\phi\}$
    \Comment{Record the current proof obligation to check the circular reasoning}
    \If{$\phi = \exists x.\,\psi$}
    \Comment{Branch over type-compatible instantiations of $x$}
        \State \Return any successful result of $\Call{Prove}{\psi^\prime, \symbAncestor^\prime, \symbSymTab}$ with $\psi^\prime \in \Call{Inst}{\psi,\symbSymTab}$
    \ElsIf{$\phi = \forall x.\,\psi$}
    \Comment{Introduce a fresh constant ${c_{\mathrm{fresh}}}$ and prove $\psi$ after mapping variable $x$ to ${c_{\mathrm{fresh}}}$}
        \State \Return $\Call{Prove}{\psi[x \mapsto {c_{\mathrm{fresh}}}], \symbAncestor^\prime, \symbSymTab \cup \{ ({c_{\mathrm{fresh}}}, \operatorname{type}_{\symbSymTab}(x), \operatorname{NL}_{\symbSymTab}(x)) \}}$
    \ElsIf{$\phi = \phi_1 \land \phi_2$}
    \Comment{Discharge both sub-obligations}
        \State \Return $\Call{Prove}{\phi_1, \symbAncestor^\prime, \symbSymTab} \cup \Call{Prove}{\phi_2, \symbAncestor^\prime, \symbSymTab}$
    \ElsIf{$\phi = \phi_1 \lor \phi_2$}
    \Comment{Discharge at least one sub-obligation}
       \State \Return the successful result of $\Call{Prove}{\phi_1, \symbAncestor^\prime, \symbSymTab}$ or $\Call{Prove}{\phi_2, \symbAncestor^\prime, \symbSymTab}$
    \EndIf
    \For{each $(\symbBridge, \symbobligation)$ yielded by $\Call{Identify}{\phi, \symbPF, \symbSymTab}$}
    \Comment{$\Call{Identify}{\phi, \symbPF, \symbSymTab}$ is given in Algorithm~\ref{alg:grounding-retrieval}}
        \State $\symbBridge_{\symbobligation} \gets \Call{Prove}{\symbobligation, \symbAncestor^\prime, \symbSymTab}$
        \If{$\symbBridge_{\symbobligation} \neq \textsc{Failed}$}
        \Comment{The residual proof obligation has been discharged}
            \State \Return $\symbBridge \cup \symbBridge_{\symbobligation}$
            \Comment{Combine the accepted bridges}
        \EndIf
    \EndFor
    \State \Return \textsc{Failed}
    \Comment{The search fails}
\EndFunction
\end{algorithmic}
\end{algorithm}

\subsubsection{Solver-Controlled Recursion}

Algorithm~\ref{alg:direct-search} defines a recursive function, $\Call{Prove}{\phi, \symbAncestor, \symbSymTab}$, that performs proof-obligation-guided search.
Here, $\phi$ denotes the current proof obligation and is initially set to $\symbgoal$; $\symbAncestor$ denotes the set of ancestor obligations to detect circular reasoning and is initially empty; and $\symbSymTab$ is the grounding table.
The solver first attempts to discharge $\phi$ directly through symbolic reasoning.
If this attempt fails, it normalizes $\phi$ and distinguishes cases based on the top-level logical structure of $\phi$.

For an existential obligation $\phi = \exists x.\,\psi$, the solver branches over the type-compatible instantiations of $x$ generated by the operator $\Call{Inst}{\psi, \symbSymTab}$, where $x$ may denote a set of variables.
More precisely, for a formula $\psi$ containing variables $x_1,\ldots,x_n$,
\begin{equation*} 
    \Call{Inst}{\psi, \symbSymTab} = \left\{ \psi[\theta] \;\middle|\; \theta = \{x_i \mapsto c_i\}_{i=1}^{n} \text{ and } \operatorname{type}_{\symbSymTab}(c_i) \preceq \operatorname{type}_{\symbSymTab}(x_i) \text{ for all } i \right\}.
\end{equation*}
Here, $\theta$ is a substitution that maps each variable $x_i$ to a constant $c_i$; $\psi[\theta]$ denotes the result of applying $\theta$ to $\psi$; and $\tau_1\preceq\tau_2$ denotes that $\tau_1$ is a subtype of $\tau_2$, such as $\mathsf{Student}\preceq\mathsf{Person}$.
Type information thus rules out type-incompatible instantiations before bridge identification begins.
For a universal obligation $\phi = \forall x.\,\psi$, the solver introduces a fresh constant ${c_{\mathrm{fresh}}}$ representing an arbitrary value of $x$ and then proves $\psi[x\mapsto{c_{\mathrm{fresh}}}]$, following the standard eigenvariable technique~\citep{lincoln1994proof}.
For a conjunctive obligation $\phi = \phi_1 \land \phi_2$, both sub-obligations must be discharged.
In contrast, for a disjunctive obligation $\phi = \phi_1 \lor \phi_2$, it suffices to discharge \mbox{at least one sub-obligation.}

\begin{algorithm}[!t]
\caption{Semantic bridge identification}
\label{alg:grounding-retrieval}
\begin{algorithmic}[1]\small
\Require An unresolved obligation $\phi$, a set of formalized premises $\symbPF$, a grounding table $\symbSymTab$, and a threshold ${\tau_{\mathrm{NLI}}}$.
\Ensure A set of accepted semantic bridge (possibly \textsc{Failed}) and its corresponding residual obligation.
\Function{Identify}{$\phi, \symbPF, \symbSymTab$}
    \State ${\mathcal{Q}_{\mathrm{bridge}}} \gets \varnothing$
    \Comment{Initialize the candidate pool}
    \ForAll{$\symbInst \subseteq \bigcup_{\symbpf \in \symbPF} \Call{Inst}{\symbpf, \symbSymTab}$}
    \Comment{Enumerate subsets of all type-compatible instantiations}
        \ForAll{${\mathbb{S}_{\mathrm{lit}}} \subseteq \bigcup_{\symbinst \in \symbInst} \Call{Lit}{\symbinst}$ with $|{\mathbb{S}_{\mathrm{lit}}}| \leq 2$}
        \Comment{Enumerate support sets comprising at most two literals}
            \State $\Sigma \gets \bigwedge_{{\ell} \in {\mathbb{S}_{\mathrm{lit}}}} {\ell}$
            \Comment{Construct the candidate semantic support}
            \State $\symbobligation \gets \Call{Abduce}{\symbInst, \Sigma}$
            \Comment{Derive the residual proof obligation}
            \State ${s_{\mathrm{sim}}} \gets \operatorname{sim} (\operatorname{NL}_{\symbSymTab}(\Sigma), \operatorname{NL}_{\symbSymTab}(\phi))$
            \Comment{Compute semantic similarity}
            \State ${\mathcal{Q}_{\mathrm{bridge}}} \gets {\mathcal{Q}_{\mathrm{bridge}}} \cup \{(\Sigma, \symbobligation, {s_{\mathrm{sim}}})\}$
        \EndFor
    \EndFor
    \For{$(\Sigma, \symbobligation, {s_{\mathrm{sim}}}) \in {\mathcal{Q}_{\mathrm{bridge}}}$ in nonincreasing order of ${s_{\mathrm{sim}}}$}
    \Comment{Evaluate candidates in ranked order}
        \State $\symbBridge \gets \varnothing$ if $\Sigma = \phi$; otherwise $\symbBridge \gets \{\Sigma \rightarrow \phi\}$
        \If{$\symbBridge = \varnothing$ \textbf{or} $\Call{NLI}{\operatorname{NL}_{\symbSymTab}(\Sigma), \operatorname{NL}_{\symbSymTab}(\phi)} \geq {\tau_{\mathrm{NLI}}}$}
        \Comment{Perform NLI filtering}
            \If{$\operatorname{Gate}(\symbBridge,\symbSymTab) = \textsc{Accepted}$}
            \Comment{Evaluate the bridge using the LLM-based $\operatorname{Gate}$}
                \State \textbf{yield} $(\symbBridge, \symbobligation)$
            \EndIf
        \EndIf
    \EndFor
\EndFunction
\end{algorithmic}
\end{algorithm}

\subsubsection{Semantic Bridge Identification}

If $\phi$ remains unresolved after the above process, Algorithm~\ref {alg:grounding-retrieval} is invoked to identify bridges that can provide semantic support for $\phi$.
The key idea is to address two questions separately: \emph{which semantic support can bridge the gap to the current proof obligation}, and \emph{what residual obligation must still be discharged to establish that support}.
The former is addressed by LLM-based evaluation, whereas the latter is addressed by symbolic reasoning.

Specifically, for each formalized premise $\symbpf \in \symbPF$, the operator $\Call{Inst}{\symbpf, \symbSymTab}$ first generates type-compatible instantiations.
For each instantiation $\symbinst$, let $\Call{Lit}{\symbinst}$ denote the set of literals it contains, where a literal is an atomic formula or its negation.
For example, the instantiation $\mathsf{Gain}(h,k) \rightarrow \mathsf{Smarter}(h)$ is equivalent to $\neg\mathsf{Gain}(h,k) \lor \mathsf{Smarter}(h)$, which contains the literals $\neg\mathsf{Gain}(h,k)$ and $\mathsf{Smarter}(h)$.
Algorithm~\ref{alg:grounding-retrieval} then considers a subset $\symbInst$ of all type-compatible instantiations and selects a support set ${\mathbb{S}_{\mathrm{lit}}} \subseteq \bigcup_{\symbinst \in \symbInst} \Call{Lit}{\symbinst}$, whose literals may come from multiple premises and may jointly support the current obligation $\phi$.
For computational tractability, ${\mathbb{S}_{\mathrm{lit}}}$ is restricted to at most two literals.
The literals in ${\mathbb{S}_{\mathrm{lit}}}$ are conjoined to form a candidate \emph{semantic support} $\Sigma = \bigwedge_{{\ell} \in {\mathbb{S}_{\mathrm{lit}}}} {\ell}$, yielding a candidate \emph{semantic bridge} $\Sigma \rightarrow \phi$.
For example, when $\phi = \mathsf{SmarterThanBefore}(h)$, selecting ${\mathbb{S}_{\mathrm{lit}}} = \{\mathsf{Smarter}(h)\}$ yields $\Sigma = \mathsf{Smarter}(h)$ and the corresponding bridge $\mathsf{Smarter}(h) \rightarrow \mathsf{SmarterThanBefore}(h)$.
Next, candidate semantic bridges are filtered using a lightweight semantic similarity model and a natural language inference~(NLI) model, both of which take $\operatorname{NL}_{\symbSymTab}(\Sigma)$ and $\operatorname{NL}_{\symbSymTab}(\phi)$ as inputs.
Finally, the LLM-based $\operatorname{Gate}$ evaluates whether each remaining bridge holds.

Even after a bridge is accepted by $\operatorname{Gate}$, its support $\Sigma$ may remain unresolved.
We therefore use logical abduction~\citep{josephson1996abductive} to derive a residual proof obligation $\symbobligation$ that, together with the instantiations in $\symbInst$, is sufficient to establish $\Sigma$:
\begin{equation*}
\symbobligation = \Call{Abduce}{\symbInst,\Sigma},
\qquad
\left(\bigwedge\nolimits_{\symbinst \in \symbInst}\symbinst\right) \land \symbobligation \models \Sigma.
\end{equation*}
The residual obligation $\symbobligation$ is then passed to Algorithm~\ref{alg:direct-search} for recursive discharge.
Thus, Algorithm~\ref{alg:grounding-retrieval} decouples semantic understanding from logical reasoning:
neural models evaluate only the candidate semantic bridge $\Sigma \rightarrow \phi$, whereas the residual obligation required to establish $\Sigma$ is derived symbolically through logical abduction and must still be discharged by the symbolic solver.

\subsubsection{Termination and Soundness}

When the search succeeds, Algorithm~\ref{alg:direct-search} returns a bridge set that satisfies both criteria in Equation~\eqref{eq:accepted-completion}:
the set enables the solver to complete the proof of goal $\symbgoal$, and every bridge in the set is accepted by $\operatorname{Gate}$.
Theoretical proof of this soundness guarantee is provided in Appendix~\ref{app:recursive-soundness}.
A search branch fails if the current obligation already appears on its ancestor path, if all candidate bridges are rejected, or if a search limit is reached.
The ancestor check aims to prevent circular reasoning~\citep{bol1991analysis}: if an attempt to prove $\phi$ generates $\psi$ and a subsequent attempt to prove $\psi$ regenerates $\phi$, the repeated obligation is rejected because it cannot serve as its own justification.

 
\section{Experimental Results}
\label{sec:exp}

\subsection{Experimental Setup}
\label{sec:exp_setup}

\textbf{Models.}
We compare \Projectname{} with Safin-1 (35B-A3B)~\cite{safin1}, the base model.
\Projectname{} incorporates REX~\cite{jin2026architecturallevelsafetycontrolrelocating} and \ReSI{} into this base model.
Post-training and \ReSI{} updates are performed on the full model, including the removable experts introduced by REX (Section~\ref{sec:rex}).

\textbf{\ReSI{} instantiation.}
We instantiate the extensible pools, training settings, and validation sets of \ReSI{} with representative resources.
The safety data pool is initialized with the WildJailbreak training split~\cite{wildteaming2024}.
The red-teaming pool contains DARWIN and JailbreakSkill~\cite{darwin2026,jailbreakskill2026}.
The alignment-method pool contains A3, SInternal, and AlphaAlign (Section~\ref{sec:resi_space}).
General-task examples are drawn from Dolci-Instruct-SFT\footnote{\url{https://huggingface.co/datasets/allenai/Dolci-Instruct-SFT}}.
All updates use LoRA with rank 32 and $\alpha=64$.
Recipe development and checkpoint selection use a development panel drawn from WildJailbreak, H-CoT, XSTest, and IFEval, which provides the validation sets for the safety gain and the Pareto gate (Sections~\ref{sec:resi_pilot} and~\ref{sec:resi_full}).

\textbf{Evaluation.}
Following \ReSI{}~\cite{resi}, we evaluate resistance, resilience, benign compliance, and retention of general capabilities, and additionally measure hazardous knowledge.
Resistance is measured by attack success rate (ASR, lower is better) on WildJailbreak, which shares its distribution with the safety data pool, and on HarmBench~\cite{harmbench2024}, which measures retention of existing defenses.
Resilience is assessed on two attack mechanisms outside the red-teaming pool: H-CoT~\cite{hcot2025}, which hijacks the model's safety reasoning with injected chain-of-thought content, and X-Teaming~\cite{xteaming2025}, an adaptive multi-turn attack.
Benign compliance is measured by the full-compliance rate on benign XSTest requests~\cite{xstest2023}.
General capabilities are measured by MMLU-Pro~\cite{mmlupro2024} and IFEval~\cite{ifeval2023}.
Hazardous knowledge is measured by multiple-choice accuracy on WMDP-Bio and WMDP-Cyber~\cite{wmdp}.
\Projectname{} and the REX-based variants in the ablation study are evaluated after excision.

\textbf{Formalization.} 
We integrate POGB into \Projectname{} as an inference-time component enabled only for logical reasoning tasks. We evaluate \Projectname{} on four logical reasoning benchmarks: FOLIO~\citep{han2024folio}, ProofWriter~\citep{tafjord2021proofwriter}, PrOntoQA~\citep{saparov2023language}, and ProverQA-Hard~\citep{qi2025large}. These benchmarks cover linguistically diverse, human-written problems and template-constrained synthetic tasks with varying reasoning complexity. We compare \Projectname{} with Gemini3-Flash~\citep{googledeepmind2025gemini3flash}, DeepSeek-V3.2~\citep{liu2025deepseek}, Qwen3.8-Flash~\citep{qiu2026design}, and GPT-5.5~\citep{openai2026gpt55}, GPT-5.6-Luna~\cite{gpt56},Claude Opus 5~\cite{claudeopus5systemcard2026},GLM-5.3~\cite{glm53model2026},Kimi K3~\cite{kimik3report2026}. each evaluated using Direct Prediction (Direct) and Chain-of-Thought (CoT) prompting. For each problem, \Projectname{} samples up to three formalizations, stopping once a valid formal proof is found and falling back to CoT if all attempts fail. For semantic bridge identification, \Projectname{} uses \texttt{all-MiniLM-L6-v2}~\citep{reimers2019sentence} to compute semantic similarity and \texttt{roberta-large-mnli}~\citep{liu2019roberta} for natural language inference. The NLI entailment threshold $\tau_{\mathrm{NLI}}$ is set to 0.80 across all datasets. We report accuracy on each benchmark and an overall average weighted by the evaluation set size of each dataset.

\subsection{Main Results}
\label{sec:exp_main}

We compare \Projectname{} with Safin-1 (Table~\ref{tab:exp_main}), place it among frontier models (Table~\ref{tab:exp_frontier}), and integrate it with formalization (Table~\ref{tab:logres}).

Results on H-CoT and X-Teaming show that \Projectname{} substantially strengthens resilience to attack mechanisms outside the red-teaming pool compared with Safin-1.
As shown in Table~\ref{tab:exp_main}, H-CoT ASR decreases from 82.00\% to 14.00\%, and X-Teaming ASR from 64.78\% to 32.08\%.

Moreover, \Projectname{} improves resistance on WildJailbreak, and its HarmBench ASR remains low.
WildJailbreak ASR decreases from 2.10\% for Safin-1 to 1.20\%, and HarmBench ASR is 1.25\%, compared with 0.25\% for Safin-1.

Compared with Safin-1, \Projectname{} retains instruction following and knowledge, with a small decrease in benign compliance.
Specifically, IFEval accuracy changes from 85.58\% to 87.06\% and MMLU-Pro accuracy from 77.10\% to 76.80\%, while benign compliance on XSTest decreases from 88.40\% to 86.80\%.

\begin{table}[!t]\small
\renewcommand{\arraystretch}{1.1}
\centering
\caption{\textbf{Comparison between Safin-1 and \Projectname{} (\%).}
WildJailbreak, HarmBench, H-CoT, and X-Teaming report ASR;
XSTest reports benign compliance; IFEval reports loose prompt-level accuracy;
WMDP reports multiple-choice accuracy.}
\label{tab:exp_main}
\setlength{\tabcolsep}{4.0pt}
\begin{tabular}{lcc}
\toprule
\textbf{Metric}
& \textbf{Safin-1}
& \textbf{\Projectname{}} \\
\midrule
WildJailbreak $\downarrow$ & 2.10 & \textbf{1.20} \\
HarmBench $\downarrow$ & \textbf{0.25} & 1.25 \\
H-CoT $\downarrow$ & 82.00 & \textbf{14.00} \\
X-Teaming $\downarrow$ & 64.78 & \textbf{32.08} \\
\midrule
XSTest $\uparrow$ & \textbf{88.40} & 86.80 \\
\midrule
MMLU-Pro $\uparrow$ & \textbf{77.10} & 76.80 \\
IFEval $\uparrow$ & 85.58 & \textbf{87.06} \\
\midrule
WMDP-Bio $\downarrow$ & 82.56 & \textbf{43.68} \\
WMDP-Cyber $\downarrow$ & 58.33 & \textbf{33.17} \\
\bottomrule
\end{tabular}
\end{table}

Compared with frontier models, \Projectname{} achieves the lowest X-Teaming and WildJailbreak ASRs and a competitive H-CoT ASR, while its benign compliance is the lowest in the table (Table~\ref{tab:exp_frontier}).
The X-Teaming ASR of \Projectname{} is 32.08\%, compared with 56.69--84.85\% for the frontier models.
Moreover, its WildJailbreak ASR of 1.20\% is lower than those of all frontier models, which range from 1.90\% to 9.25\%.
On H-CoT, \Projectname{} reaches 14.00\%, lower than Claude Opus 5, GLM-5.3, and Kimi K3 and higher than GPT-5.6-Luna.
However, its benign compliance of 86.80\% is lower than those of all other models in the table.

\begin{table}[!t]\small
\renewcommand{\arraystretch}{1.1}
\centering
\caption{\textbf{Comparison with frontier models (\%).}
WJB and HB denote WildJailbreak and HarmBench.
Safety columns report ASR; XSTest reports benign compliance.}
\label{tab:exp_frontier}
\setlength{\tabcolsep}{4.0pt}
\begin{tabular}{lccccc}
\toprule
\textbf{Model}
& \textbf{WJB} $\downarrow$
& \textbf{HB} $\downarrow$
& \textbf{H-CoT} $\downarrow$
& \textbf{X-Teaming} $\downarrow$
& \textbf{XSTest} $\uparrow$ \\
\midrule
GPT-5.6-Luna~\cite{gpt56}
& 1.90 & \textbf{0.00} & \textbf{8.00} & 56.69 & \textbf{100.00} \\
Claude Opus 5~\cite{claudeopus5systemcard2026}
& 9.25 & 3.50 & 38.00 & 84.85 & 98.80 \\
GLM-5.3~\cite{glm53model2026}
& 5.35 & 3.25 & 24.00 & 76.73 & 99.20 \\
Kimi K3~\cite{kimik3report2026}
& 4.50 & 3.75 & 18.00 & 77.80 & 99.60 \\
\midrule
Safin-1
& 2.10 & 0.25 & 82.00 & 64.78 & 88.40 \\
\Projectname{}
& \textbf{1.20} & 1.25 & 14.00 & \textbf{32.08} & 86.80 \\
\bottomrule
\end{tabular}
\end{table}
As shown in Table~\ref{tab:logres}, \Projectname{} combines strong predictive accuracy with verifiable reasoning. On FOLIO, it achieves the highest accuracy of 88.73\%, surpassing the strongest baseline, Claude Opus 5 with Direct prompting, by 0.98 percentage points. It also attains 100.00\% on PrOntoQA, matching the best-performing baselines, and near-perfect accuracy of 99.50\% on ProofWriter, only 0.33 percentage points below the best result. On ProverQA-Hard, \Projectname{} achieves 75.60\% accuracy. Across all 1,804 evaluation examples, its overall accuracy reaches 91.80\%, exceeding every evaluated baseline configuration except those using Claude Opus 5. Beyond answer accuracy, \Projectname{} provides machine-checkable proofs whenever formal proof search succeeds. These proofs allow the corresponding conclusions to be independently verified from the formalized premises and accepted semantic bridges, providing an explicit verification mechanism alongside competitive end-to-end performance.

\begin{table}[!t]\small
    \renewcommand{\arraystretch}{1.1}
    \centering
    \caption{\textbf{End-to-end reasoning accuracy (\%).} The evaluation set size ($N$) for each dataset is shown below its column header.}
    \label{tab:logres}
    \setlength{\tabcolsep}{4.0pt}
    \begin{tabular}{llccccc}
    \toprule
    \textbf{LLM} & \textbf{Method}
    & \makecell{\textbf{FOLIO} $\uparrow$\\\textcolor{gray}{$N = 204$}}
    & \makecell{\textbf{ProofWriter} $\uparrow$\\\textcolor{gray}{$N = 600$}}
    & \makecell{\textbf{PrOntoQA} $\uparrow$\\\textcolor{gray}{$N = 500$}}
    & \makecell{\textbf{ProverQA} $\uparrow$\\\textcolor{gray}{$N = 500$}}
    & \textbf{Average} $\uparrow$ \\
    \midrule
    \multirow{2}{*}{{Gemini3-Flash}}
    & Direct & 69.61 & 77.50 & 95.80 & 65.00 & 78.22 \\
    & CoT    & 83.82 & 95.33 & \textbf{100.00} & 74.80 & 89.63 \\
    \midrule
    \multirow{2}{*}{{DeepSeek-V3.2}}
    & Direct & 65.69 & 57.67 & 68.20 & 58.80 & 61.81 \\
    & CoT    & 84.80 & 91.67 & \textbf{100.00} & 69.40 & 87.03 \\
    \midrule
    \multirow{2}{*}{{Qwen3.8-Flash}}
    & Direct & 68.14 & 72.83 & 90.80 & 64.00 & 74.83 \\
    & CoT    & 82.35 & 94.33 & \textbf{100.00} & 73.40 & 88.75 \\
    \midrule
    \multirow{2}{*}{{GPT-5.5}}
    & Direct & 77.45 & 98.67 & \textbf{100.00} & 67.80 & 88.08 \\
    & CoT    & 79.41 & 98.50 & \textbf{100.00} & 66.40 & 87.86 \\
    \midrule
    \multirow{2}{*}{{GPT-5.6-Luna}}
    & Direct & 62.75 & 95.83 & 99.80 & 48.60 & 80.10 \\
    & CoT    & 74.02 & \textbf{99.83} & \textbf{100.00} & 51.80 & 83.65 \\
    \midrule
    \multirow{2}{*}{{Claude Opus 5}}
    & Direct & 87.75 & 98.50 & \textbf{100.00} & \textbf{84.00} & \textbf{93.68} \\
    & CoT    & 82.35 & 98.17 & 99.80 & 81.80 & 92.29 \\
    \midrule
    \multirow{2}{*}{{GLM-5.3}}
    & Direct & 81.37 & 98.50 & \textbf{100.00} & 62.00 & 86.86 \\
    & CoT    & 82.84 & 98.83 & \textbf{100.00} & 69.80 & 89.30 \\
    \midrule
    \multirow{2}{*}{{Kimi K3}}
    & Direct & 83.33 & 98.50 & \textbf{100.00} & 76.20 & 91.02 \\
    & CoT    & 83.82 & 98.17 & \textbf{100.00} & 76.00 & 90.91 \\
    \midrule
    {\Projectname{}}
    & N/A & \textbf{88.73} & 99.50 & \textbf{100.00} & 75.60 & 91.80 \\
    \bottomrule
    \end{tabular}
\end{table}

\subsection{Analysis}
\label{sec:exp_analysis}

\subsubsection{Safety-Native Architecture}
\label{sec:exp_moe}

We first evaluate REX itself on Safin-1-35B-A3B, against the same RMU recipe applied directly to the original model.
$R$ applies the RMU operator, hyperparameters, and forget budget of the residual cleanup directly to $B$, so that $R$ and $A_2$ differ only in their starting point and retain reference.
For post-training, three routes use the identical 250M-token standard full-parameter SFT with the same data order and random seed, and differ only in their starting points.
$A_3$, which is Safin-1+REX, is trained on the recombined full model as $F_3$ and excised at deployment; $A_3$-direct is trained directly on the excised $A_2$; and $R$-SFT is trained on $R$.
Hazardous knowledge is measured by WMDP-Bio and WMDP-Cyber, both four-way multiple choice with a chance level of 25\%, and general knowledge by MMLU.
Results are listed in Table~\ref{tab:exp_rex}.

\begin{table}[!t]\small
\renewcommand{\arraystretch}{1.1}
\centering
\caption{\textbf{REX versus the same RMU recipe on Safin-1-35B-A3B (\%).}
Kept share is $(W_B - W_{250\mathrm{M}})/(W_B - W_0)$,
the fraction of the removal at the start of post-training that remains
after 250M tokens, where $W$ denotes WMDP accuracy.
}
\label{tab:exp_rex}
\setlength{\tabcolsep}{4.0pt}
\begin{tabular}{llcccc}
\toprule
\textbf{Stage} & \textbf{Model}
& \textbf{WMDP-Bio} $\downarrow$
& \textbf{WMDP-Cyber} $\downarrow$
& \textbf{MMLU} $\uparrow$
& \makecell{\textbf{Kept share} $\uparrow$\\
\textbf{(Bio / Cyber)}} \\
\midrule
Original & $B$ (Safin-1)
& 82.6 & 58.3 & 82.3 & -- \\
\midrule
\multirow{2}{*}{{Base}}
& $R$
& 56.9 & 32.7 & 82.0 & -- \\
& $A_2$ (REX)
& \textbf{36.2} & \textbf{29.5} & 78.9 & -- \\
\midrule
\multirow{4}{*}{\makecell[l]{{250M}\\{post-training}}}
& $R$-SFT
& 81.7 & 47.2 & 82.4 & 3.4 / 43.6 \\
& $A_3$-direct
& 64.0 & 41.8 & 80.7 & 40.0 / 57.5 \\
& $A_3$ (REX)
& \textbf{40.8} & \textbf{32.5} & 80.2
& \textbf{90.0 / 89.9} \\
& $F_3$ (full model)
& 82.8 & 57.4 & 82.6 & -- \\
\bottomrule
\end{tabular}
\end{table}

Excising the removable experts takes the target knowledge of both domains away, while the full model retains it.
After 250M tokens of post-training, the full model $F_3$ reaches 82.8\% on WMDP-Bio and 57.4\% on WMDP-Cyber, on par with 82.6\% and 58.3\% for $B$.
Its excised counterpart $A_3$ falls to 40.8\% and 32.5\%, with Cyber close to chance level, while MMLU differs by only 2.4 points.
The same post-training therefore produces two deliverables at once: a full model for authorized users that retains the hazardous capability, and a public model that contains no parameters of the removable experts.
A single set of removable experts serves both domains, and one excision removes the target knowledge of both, so the removable experts act as a common removal unit for multiple hazardous domains.

Before post-training, the public base model $A_2$ of REX is already lower than $R$ in both domains under the same RMU recipe.
$A_2$ reaches 36.2\% on WMDP-Bio, 20.7 points below the 56.9\% of $R$, and 29.5\% on WMDP-Cyber, closer to chance level than the 32.7\% of $R$.
The two models share an identical RMU, so the difference comes from REX first relocating the target knowledge into the removable experts and excising it, which lets the cleanup start from a lower point.
$A_2$ retains over 96\% of the MMLU accuracy of $B$.

Under the same 250M tokens of standard post-training, the removal by the same RMU recipe almost entirely disappears, whereas the public model of REX stays low, which is the most direct value of REX for Safety RSI.
$R$-SFT returns to 81.7\% on WMDP-Bio, nearly the 82.6\% of $B$, while $A_3$ stays at 40.8\%, 40.9 points lower.
On WMDP-Cyber, $R$-SFT returns to 47.2\% and $A_3$ stays at 32.5\%.
The kept share measures how much of the removal at the start of post-training survives it, independent of where each route starts.
On Bio, $A_3$ keeps 90.0\% of the removal present at the start, whereas $R$-SFT keeps only 3.4\%; on Cyber, the two keep 89.9\% and 43.6\%.
$A_3$-direct starts from the same $A_2$ but is trained on the public form without recombination, and its WMDP-Bio rises to 64.0\%, a kept share of 40.0\%.
Since the three trajectories share data, budget, and order, the gap between $R$-SFT and $A_3$-direct comes from relocation and excision, and the gap between $A_3$-direct and $A_3$ comes from training on the full structure and excising at deployment.
The latter gap also directly supports the choice in \ReSI{} to train the full model rather than its excised counterpart (Section~\ref{sec:resi_arch}).
Together, the two effects keep the public model close to its pre-post-training removal level, while general knowledge continues to grow, with MMLU rising from 78.9\% for $A_2$ to 80.2\% for $A_3$.

The same delivery scheme carries over to the safety updates of \ReSI{}.
Starting from $F_3$, \ReSI{} trains the full model with all experts and red-teams, validates, and gates its excised counterpart (Section~\ref{sec:resi_arch}).
After this round, the excised model stays at 43.68\% on WMDP-Bio and 33.17\% on WMDP-Cyber, compared with 82.56\% and 58.33\% for Safin-1 (Table~\ref{tab:exp_main}).
Although \ReSI{} differs from standard post-training in both its objective and its data, the removal brought by REX carries through this round as well.
Updating on the full structure and excising at deployment can thus serve as a fixed interface for every round of Safety RSI, so that neither safety repairs nor capability gains come at the cost of hazardous knowledge returning to the public model.

\subsubsection{ReSI: Recursive Safety Improvement}
\label{sec:exp_rsi}

Following \ReSI{}~\cite{resi}, we analyze the results from Safin-1 and Safin-1+REX in terms of resistance, resilience, and retention of general capabilities and benign compliance (Table~\ref{tab:exp_a3_comparison}); Table~\ref{tab:exp_main} presents the overall comparison between Safin-1 and \Projectname{}.

\textbf{Resistance: in-distribution safety improvements and retention of existing defenses.}
Results on WildJailbreak and HarmBench illustrate improved in-distribution safety and retention of existing safety defenses, respectively, from both starting models after \ReSI{} alignment.
Safin-1 improves from 2.10\% to 0.35\% on WildJailbreak and retains its already low HarmBench ASR of 0.25\%.
Safin-1+REX improves from 2.10\% to 1.20\% on WildJailbreak, and its HarmBench ASR remains low, changing from 1.00\% to 1.25\%.

\textbf{Resilience: generalization to unseen attack mechanisms.}
Beyond in-distribution improvements, \ReSI{} substantially strengthens safety against attack mechanisms outside the red-teaming pool.
H-CoT ASRs decrease from 82.00\% and 64.00\% before alignment to 6.00\% and 14.00\% afterward for Safin-1 and Safin-1+REX, respectively.
In the overall comparison, X-Teaming ASR decreases from 64.78\% for Safin-1 to 32.08\% for \Projectname{} (Table~\ref{tab:exp_main}).

\textbf{Retention of general capabilities and benign compliance.}
The safety gains from \ReSI{} accompany broadly retained general capabilities and benign compliance, with variation across starting models.
After alignment, benign-compliance rates on XSTest decrease by 2.80 percentage points for Safin-1 and by 8.40 points for Safin-1+REX.
MMLU-Pro accuracy changes by $+2.30$ and $-2.40$ percentage points, and IFEval loose prompt-level accuracy by $-1.29$ and $-5.55$ points, for Safin-1 and Safin-1+REX, respectively.

\textbf{Comparison with the A3 baseline.}
Table~\ref{tab:exp_a3_comparison} compares A3~\cite{a3_2026}, a fixed alignment method with its own data construction and refinement, with \ReSI{} on both starting models.
On Safin-1, \ReSI{} achieves lower ASRs and better benign compliance and capability retention than A3 across the reported metrics.
On Safin-1+REX, \ReSI{} yields \Projectname{}, which matches A3 on HarmBench and H-CoT and improves WildJailbreak, XSTest, and IFEval, although MMLU-Pro accuracy is lower.
The comparison highlights the model-dependent tradeoff between safety, benign compliance, and general capabilities.

\begin{table}[!t]\small
\renewcommand{\arraystretch}{1.1}
\centering
\caption{\textbf{Comparison with A3 on both starting models (\%).}
WJB and HB denote WildJailbreak and HarmBench;
XSTest reports benign compliance, and IFEval reports loose prompt-level accuracy.
The final $+$\ReSI{} row corresponds to \Projectname{}.}
\label{tab:exp_a3_comparison}
\setlength{\tabcolsep}{4.0pt}
\begin{tabular}{lcccccc}
\toprule
\textbf{Model}
& \textbf{WJB} $\downarrow$
& \textbf{HB} $\downarrow$
& \textbf{H-CoT} $\downarrow$
& \textbf{XSTest} $\uparrow$
& \textbf{MMLU-Pro} $\uparrow$
& \textbf{IFEval} $\uparrow$ \\
\midrule
Safin-1
& 2.10 & 0.25 & 82.00 & 88.40 & 77.10 & 85.58 \\
\quad + A3
& 0.70 & 1.00 & 24.00 & 84.80 & 77.40 & 80.59 \\
\quad + \ReSI{}
& 0.35 & 0.25 & 6.00 & 85.60 & 79.40 & 84.29 \\
\midrule
Safin-1+REX
& 2.10 & 1.00 & 64.00 & 95.20 & 79.20 & 92.61 \\
\quad + A3
& 1.40 & 1.25 & 14.00 & 84.80 & 79.80 & 72.46 \\
\quad + \ReSI{}
& 1.20 & 1.25 & 14.00 & 86.80 & 76.80 & 87.06 \\
\bottomrule
\end{tabular}
\end{table}

\FloatBarrier

\section{Conclusions}
In this work, we presented \Projectname{}, a model built on a safety architecture, trained through recursive safety improvement, and equipped with an inference-time formal reasoning harness. Persistent safety memory and hazardous-knowledge disentanglement support the retention of safety knowledge and the control of hazardous capabilities, while recursive safety improvement turns newly discovered vulnerabilities into evaluated repairs. Experiments show substantial reductions in attack success rates while largely preserving general capabilities and benign compliance. The formal reasoning harness further improves reasoning accuracy by combining the model's language understanding with symbolic inference and verification, producing checkable proofs under formalized premises and accepted semantic bridges. These results demonstrate progress toward the AI-45° Law's goal of coordinated safety and capability development. In the future, we plan to examine safety retention over longer improvement trajectories and explore formal reasoning for safety decisions governed by explicit rules.







\section*{Authorship Statement}

This work was led by the Shanghai Artificial Intelligence Laboratory with contributions from the authors listed below. 

\begin{table}[htbp]
    \centering
    \renewcommand{\arraystretch}{1.35}
    \setlength{\tabcolsep}{6pt}
    \begin{tabularx}{\textwidth}{
        @{}
        *{4}{>{\centering\arraybackslash}X}
        @{}
    }
        Jingnan Zheng\textsuperscript{*} & Cheng Jin\textsuperscript{*} & Ming Zhang\textsuperscript{*} & Huan Sun\textsuperscript{*} \\

        Zhongtian Ma & Kaisen Yang  & Ziyuan Zhou & Tong Wu \\

        Dongcheng Zhang & Yi Zhang & Xiao Wang & Han Qi  \\

        Shu Yu & Ermo Hua & Zongji Yu & Chaozhuo Li \\

        Xiaoyu Wen & Zhekai Chen & Jiawei Zhou & Sirui Chen \\
        
        Ning Ding & An Zhang & Tat-Seng Chua & Xiangnan He \\
        
        Xia Hu & Bowen Zhou & Qiaosheng Zhang\textsuperscript{\ddag} & Youbang Sun\textsuperscript{\ddag} \\

        Xiang Wang\textsuperscript{\ddag} & Zhongyu Huang\textsuperscript{\ddag} & Chaochao Lu\textsuperscript{\dag} \\
    \end{tabularx}

    \smallskip
    {\small $^{*}$\textbf{Equal Contribution}\quad $^{\ddag}$\textbf{Technical Lead}\quad $^{\dag}$\textbf{Project Lead}}
\end{table}

\begingroup
\sloppy
\printbibliography[heading=bibintoc]

@string {N = "Nature"}

@string {C = "Cell"}

@string {ICML = "International Conference on Machine Learning (ICML)"}

@string {EMNLP = "Annual Conference on Empirical Methods in Natural Language Processing (EMNLP)"}

@string {CHI = "ACM Conference on Human Factors in Computing Systems (CHI)"}

@inproceedings{zhang2026darwin,
    title={Darwin G\"odel Machine: Open-Ended Evolution of Self-Improving Agents},
    author={Jenny Zhang and Shengran Hu and Cong Lu and Robert Tjarko Lange and Jeff Clune},
    booktitle={The Fourteenth International Conference on Learning Representations},
    year={2026},
    url={https://openreview.net/forum?id=pUpzQZTvGY}
}

@misc{lu2024aiscientistfullyautomated,
      title={The AI Scientist: Towards Fully Automated Open-Ended Scientific Discovery}, 
      author={Chris Lu and Cong Lu and Robert Tjarko Lange and Jakob Foerster and Jeff Clune and David Ha},
      year={2024},
      eprint={2408.06292},
      archivePrefix={arXiv},
      primaryClass={cs.AI},
      url={https://arxiv.org/abs/2408.06292}, 
}

@misc{gpt56,
      title={{GPT}-5.6 System Card}, 
      author={OpenAI},
      year={2026},
      howpublished={OpenAI Deployment Safety Hub},
      url={https://deploymentsafety.openai.com/gpt-5-6},
      note={Accessed 20 September 2026}
}

@misc{claude4,
      title={Claude 4 System Card}, 
      author={Anthropic},
      year={2025},
      url={https://www-cdn.anthropic.com/6d8a8055020700718b0c49369f60816ba2a7c285.pdf}, 
}

@inproceedings{harmbench,
    author = {Mazeika, Mantas and Phan, Long and Yin, Xuwang and Zou, Andy and Wang, Zifan and Mu, Norman and Sakhaee, Elham and Li, Nathaniel and Basart, Steven and Li, Bo and Forsyth, David and Hendrycks, Dan},
    title = {HarmBench: a standardized evaluation framework for automated red teaming and robust refusal},
    year = {2024},
    publisher = {JMLR.org},
    booktitle = {Proceedings of the 41st International Conference on Machine Learning},
    articleno = {1431},
    numpages = {44},
    location = {Vienna, Austria},
    series = {ICML'24}
}

@inproceedings{qi2025safety,
    title={Safety Alignment Should be Made More Than Just a Few Tokens Deep},
    author={Xiangyu Qi and Ashwinee Panda and Kaifeng Lyu and Xiao Ma and Subhrajit Roy and Ahmad Beirami and Prateek Mittal and Peter Henderson},
    booktitle={The Thirteenth International Conference on Learning Representations},
    year={2025},
    url={https://openreview.net/forum?id=6Mxhg9PtDE}
}

@InProceedings{wmdp,
    title = 	 {The {WMDP} Benchmark: Measuring and Reducing Malicious Use with Unlearning},
    author =       {Li, Nathaniel and Pan, Alexander and Gopal, Anjali and Yue, Summer and Berrios, Daniel and Gatti, Alice and Li, Justin D. and Dombrowski, Ann-Kathrin and Goel, Shashwat and Mukobi, Gabriel and Helm-Burger, Nathan and Lababidi, Rassin and Justen, Lennart and Liu, Andrew Bo and Chen, Michael and Barrass, Isabelle and Zhang, Oliver and Zhu, Xiaoyuan and Tamirisa, Rishub and Bharathi, Bhrugu and Herbert-Voss, Ariel and Breuer, Cort B and Zou, Andy and Mazeika, Mantas and Wang, Zifan and Oswal, Palash and Lin, Weiran and Hunt, Adam Alfred and Tienken-Harder, Justin and Shih, Kevin Y. and Talley, Kemper and Guan, John and Steneker, Ian and Campbell, David and Jokubaitis, Brad and Basart, Steven and Fitz, Stephen and Kumaraguru, Ponnurangam and Karmakar, Kallol Krishna and Tupakula, Uday and Varadharajan, Vijay and Shoshitaishvili, Yan and Ba, Jimmy and Esvelt, Kevin M. and Wang, Alexandr and Hendrycks, Dan},
    booktitle = 	 {Proceedings of the 41st International Conference on Machine Learning},
    pages = 	 {28525--28550},
    year = 	 {2024},
    editor = 	 {Salakhutdinov, Ruslan and Kolter, Zico and Heller, Katherine and Weller, Adrian and Oliver, Nuria and Scarlett, Jonathan and Berkenkamp, Felix},
    volume = 	 {235},
    series = 	 {Proceedings of Machine Learning Research},
    month = 	 jul,
    publisher =    {PMLR},
    url = 	 {https://proceedings.mlr.press/v235/li24bc.html}
}

@inproceedings{zhang2024negative,
    title={Negative Preference Optimization: From Catastrophic Collapse to Effective Unlearning},
    author={Ruiqi Zhang and Licong Lin and Yu Bai and Song Mei},
    booktitle={First Conference on Language Modeling},
    year={2024},
    url={https://openreview.net/forum?id=MXLBXjQkmb}
}

@article{yao2024survey,
  title={A survey on large language model (llm) security and privacy: The good, the bad, and the ugly},
  author={Yao, Yifan and Duan, Jinhao and Xu, Kaidi and Cai, Yuanfang and Sun, Zhibo and Zhang, Yue},
  journal={High-Confidence Computing},
  volume={4},
  number={2},
  pages={100211},
  year={2024},
  publisher={Elsevier}
}

@misc{safeworkr1,
      title={SafeWork-R1: Coevolving Safety and Intelligence under the AI-45$^{\circ}$ Law}, 
      author={Shanghai AI Lab and : and Yicheng Bao and Guanxu Chen and Mingkang Chen and Yunhao Chen and Chiyu Chen and Lingjie Chen and Sirui Chen and Xinquan Chen and Jie Cheng and Yu Cheng and Dengke Deng and Yizhuo Ding and Dan Ding and Xiaoshan Ding and Yi Ding and Zhichen Dong and Lingxiao Du and Yuyu Fan and Xinshun Feng and Yanwei Fu and Yuxuan Gao and Ruijun Ge and Tianle Gu and Lujun Gui and Jiaxuan Guo and Qianxi He and Yuenan Hou and Xuhao Hu and Hong Huang and Kaichen Huang and Shiyang Huang and Yuxian Jiang and Shanzhe Lei and Jie Li and Lijun Li and Hao Li and Juncheng Li and Xiangtian Li and Yafu Li and Lingyu Li and Xueyan Li and Haotian Liang and Dongrui Liu and Qihua Liu and Zhixuan Liu and Bangwei Liu and Huacan Liu and Yuexiao Liu and Zongkai Liu and Chaochao Lu and Yudong Lu and Xiaoya Lu and Zhenghao Lu and Qitan Lv and Caoyuan Ma and Jiachen Ma and Xiaoya Ma and Zhongtian Ma and Lingyu Meng and Ziqi Miao and Yazhe Niu and Yuezhang Peng and Yuan Pu and Han Qi and Chen Qian and Xingge Qiao and Jingjing Qu and Jiashu Qu and Wanying Qu and Wenwen Qu and Xiaoye Qu and Qihan Ren and Qingnan Ren and Qingyu Ren and Jing Shao and Wenqi Shao and Shuai Shao and Dongxing Shi and Xin Song and Xinhao Song and Yan Teng and Xuan Tong and Yingchun Wang and Xuhong Wang and Shujie Wang and Xin Wang and Yige Wang and Yixu Wang and Yuanfu Wang and Futing Wang and Ruofan Wang and Wenjie Wang and Yajie Wang and Muhao Wei and Xiaoyu Wen and Fenghua Weng and Yuqi Wu and Yingtong Xiong and Xingcheng Xu and Chao Yang and Yue Yang and Yang Yao and Yulei Ye and Zhenyun Yin and Yi Yu and Bo Zhang and Qiaosheng Zhang and Jinxuan Zhang and Yexin Zhang and Yinqiang Zheng and Hefeng Zhou and Zhanhui Zhou and Pengyu Zhu and Qingzi Zhu and Yubo Zhu and Bowen Zhou},
      year={2025},
      eprint={2507.18576},
      archivePrefix={arXiv},
      primaryClass={cs.AI},
      url={https://arxiv.org/abs/2507.18576}, 
}

@article{chen2026recursive,
  title={Recursive Self-Improvement in AI: From Bounded Self-Refinement to Autonomous Research Loops},
  author={Chen, Mingguang and Wang, Licheng and Qu, Bo},
  journal={arXiv preprint arXiv:2607.07663},
  year={2026}
}

@article{ibrahimouglu2026autonomous,
  title={Autonomous AI agents and recursive self-improvement: a comprehensive survey of contemporary architectures, capabilities, and safety frameworks},
  author={{\.I}brahimo{\u{g}}lu, Nadir and Demirezen, Mustafa Umut and Y{\i}ld{\i}z, Furkan},
  journal={Cluster Computing},
  volume={29},
  number={13},
  pages={766},
  year={2026},
  publisher={Springer}
}

@article{gong2026evolutionary,
  title={Evolutionary Safety of Recursive Self-Improving AI: Taxonomy, Risk Discovery, and Evaluation},
  author={Gong, Chang and Bi, Jingping and Yao, Di and Liang, Xinjian and Xiang, Chao and Guo, Ruijie},
  journal={arXiv preprint arXiv:2609.31186},
  year={2026}
}

@article{zhang2026safety,
  title={Safety Must Survive Self-Improvement: Why Failures Persist and How Agents Recover},
  author={Zhang, Yunbei and Wang, Janet and Lyu, Saiyue and Ge, Yingqiang and Liang, Kaiqu and Jin, Zijian and Reddy, Chandan K and Hamm, Jihun},
  journal={arXiv preprint arXiv:2610.01073},
  year={2026}
}

@inproceedings{qi2024fine,
  title={Fine-tuning aligned language models compromises safety, even when users do not intend to!},
  author={Qi, Xiangyu and Zeng, Yi and Xie, Tinghao and Chen, Pin-Yu and Jia, Ruoxi and Mittal, Prateek and Henderson, Peter},
  booktitle={International Conference on Learning Representations},
  volume={2024},
  pages={30988--31043},
  year={2024}
}

@article{doshi2024does,
  title={Does unlearning truly unlearn? a black box evaluation of llm unlearning methods},
  author={Doshi, Jai and Stickland, Asa Cooper},
  journal={arXiv preprint arXiv:2411.12103},
  year={2024}
}

@article{bol1991analysis,
  title = {An analysis of loop checking mechanisms for logic programs},
  author = {Bol, Roland N and Apt, Krzysztof R and Klop, Jan Willem},
  journal = {Theoretical Computer Science},
  volume = {86},
  number = {1},
  pages = {35--79},
  year = {1991}
}

@inproceedings{lincoln1994proof,
  title = {Proof search in first-order linear logic and other cut-free sequent calculi},
  author = {Lincoln, PD and Shankar, Natarajan},
  booktitle = {Proceedings of the 9th Annual Symposium on Logic in Computer Science},
  pages = {282--291},
  year = {1994}
}

@book{josephson1996abductive,
  title = {Abductive inference: Computation, philosophy, technology},
  author = {Josephson, John R and Josephson, Susan G},
  year = {1996},
  publisher = {Cambridge University Press}
}

@article{cardelli1996type,
  title = {Type systems},
  author = {Cardelli, Luca},
  journal = {ACM Computing Surveys},
  volume = {28},
  number = {1},
  pages = {263--264},
  year = {1996}
}

@inproceedings{han2024folio,
  title = {{FOLIO}: Natural language reasoning with first-order logic},
  author = {Han, Simeng and Schoelkopf, Hailey and Zhao, Yilun and Qi, Zhenting and Riddell, Martin and Zhou, Wenfei and Coady, James and Peng, David and Qiao, Yujie and Benson, Luke and others},
  booktitle = {Proceedings of the 2024 Conference on Empirical Methods in Natural Language Processing},
  pages = {22017--22031},
  year = {2024}
}

@inproceedings{tafjord2021proofwriter,
  title = {{ProofWriter}: Generating implications, proofs, and abductive statements over natural language},
  author = {Tafjord, Oyvind and Dalvi, Bhavana and Clark, Peter},
  booktitle = {Findings of the Association for Computational Linguistics: ACL-IJCNLP 2021},
  pages = {3621--3634},
  year = {2021}
}

@inproceedings{saparov2023language,
  title = {Language models are greedy reasoners: A systematic formal analysis of chain-of-thought},
  author = {Saparov, Abulhair and He, He},
  booktitle = {Proceedings of the 11th International Conference on Learning Representations},
  year = {2023}
}

@inproceedings{qi2025large,
  title = {Large language models meet symbolic provers for logical reasoning evaluation},
  author = {Qi, Chengwen and Ma, Ren and Li, Bowen and Hui, Binyuan and Wu, Jinwang and Laili, Yuanjun and He, Conghui and others},
  booktitle = {Proceedings of the 13th International Conference on Learning Representations},
  pages = {54186--54205},
  year = {2025}
}

@misc{safin1,
      title={Safin-1: Safety from Within through Memory-Native State Evolution}, 
      author={Ming Zhang and Kaisen Yang and Shu Yu and Ermo Hua and Zhekai Chen and Cheng Jin and Jingnan Zheng and Yi Zhang and Zhongtian Ma and Jiawei Zhou and Sirui Chen and Qiaosheng Zhang and Xiang Wang and Ning Ding and Xia Hu and Bowen Zhou and Youbang Sun and Chaochao Lu},
      year={2026},
      eprint={2609.00092},
      archivePrefix={arXiv},
      primaryClass={cs.LG},
      url={https://arxiv.org/abs/2609.00092}, 
}

@misc{googledeepmind2025gemini3flash,
  title = {{Gemini 3 Flash - Model} card},
  author = {{Google DeepMind}},
  url = {https://deepmind.google/models/model-cards/gemini-3-flash/},
  year = {2025}
}

@article{liu2025deepseek,
  title = {{DeepSeek-V3.2}: Pushing the frontier of open large language models},
  author = {Liu, Aixin and Mei, Aoxue and Lin, Bangcai and Xue, Bing and Wang, Bingxuan and Xu, Bingzheng and Wu, Bochao and Zhang, Bowei and Lin, Chaofan and Dong, Chen and others},
  journal = {arXiv preprint arXiv:2512.02556},
  year = {2025}
}

@article{qiu2026design,
  title = {On the design of {Qwen3.8-Next} architecture: Evaluation, efficiency, and training stability},
  author = {Qiu, Zihan and Wang, Zekun and Li, Xiao and Li, Yanpeng and Xu, Yang and Wang, Yixuan and Zhang, Huaqing and Men, Rui and Mao, Bochao and Zhang, Chengruidong and others},
  journal = {arXiv preprint arXiv:2608.30320},
  year = {2026}
}

@misc{openai2026gpt55,
  title = {{GPT-5.5} system card},
  author = {{OpenAI}},
  url = {https://deploymentsafety.openai.com/gpt-5-5},
  year = {2026}
}

@inproceedings{reimers2019sentence,
  title = {{Sentence-BERT}: Sentence embeddings using siamese {BERT}-networks},
  author = {Reimers, Nils and Gurevych, Iryna},
  booktitle = {Proceedings of the 2019 Conference on Empirical Methods in Natural Language Processing and the 9th International Joint Conference on Natural Language Processing},
  pages = {3982--3992},
  year = {2019}
}

@article{liu2019roberta,
  title = {{RoBERTa}: A robustly optimized {BERT} pretraining approach},
  author = {Liu, Yinhan and Ott, Myle and Goyal, Naman and Du, Jingfei and Joshi, Mandar and Chen, Danqi and Levy, Omer and Lewis, Mike and Zettlemoyer, Luke and Stoyanov, Veselin},
  journal = {arXiv preprint arXiv:1907.11692},
  year = {2019}
}

@inproceedings{sunposition,
  title={Position: Safe {AI} Should be Resistant and Resilient in an Evolving World},
  author={Sun, Youbang and Wang, Xiang and Fu, Jie and Lu, Chaochao and Zhou, Bowen},
  booktitle={Forty-third International Conference on Machine Learning Position Paper Track},
  year={2026},
  url={https://openreview.net/forum?id=ApUb0ToKMN},
  eprint={2509.06786},
  archiveprefix={arXiv},
  primaryclass={cs.LG}
}

@article{wang2026lifelong,
  title={Lifelong safety alignment for language models},
  author={Wang, Haoyu and Zhao, Yifei and Qin, Zeyu and Du, Chao and Lin, Min and Wang, Xueqian and Pang, Tianyu},
  journal={Advances in Neural Information Processing Systems},
  volume={38},
  pages={19908--19938},
  year={2026}
}

@misc{GPT-Red,
      title={GPT-Red: Automated Red Teaming via Self-Play at Scale}, 
      author={Eric Wallace and Christopher A. Choquette-Choo and Nikhil Kandpal and Sam Toyer and Dylan Hunn and Stephanie Lin and Yuxin Wen and Xiangyu Qi and Christopher Wolff and Zizhao Wang and Milad Nasr and Sicheng Zhu and Chuan Guo and Juan Felipe Cerón Uribe and Kaiwen Wang and Aiden Low and Kai Xiao and Kai Chen},
      year={2026},
      eprint={2607.26115},
      archivePrefix={arXiv},
      primaryClass={cs.CR},
      url={https://arxiv.org/abs/2607.26115}, 
}

@article{10.1145/1592434.1592436,
    author = {Woodcock, Jim and Larsen, Peter Gorm and Bicarregui, Juan and Fitzgerald, John},
    title = {Formal methods: Practice and experience},
    year = {2009},
    issue_date = {October 2009},
    publisher = {Association for Computing Machinery},
    address = {New York, NY, USA},
    volume = {41},
    number = {4},
    issn = {0360-0300},
    url = {https://doi.org/10.1145/1592434.1592436},
    doi = {10.1145/1592434.1592436},
    journal = {ACM Comput. Surv.},
    month = oct,
    articleno = {19},
    numpages = {36}
}

@inproceedings{pan-etal-2023-logic,
    title = "Logic-{LM}: Empowering Large Language Models with Symbolic Solvers for Faithful Logical Reasoning",
    author = "Pan, Liangming  and
      Albalak, Alon  and
      Wang, Xinyi  and
      Wang, William",
    editor = "Bouamor, Houda  and
      Pino, Juan  and
      Bali, Kalika",
    booktitle = "Findings of the {A}ssociation for {C}omputational {L}inguistics: {EMNLP} 2023",
    month = dec,
    year = "2023",
    address = "Singapore",
    publisher = "Association for Computational Linguistics",
    url = "https://aclanthology.org/2023.findings-emnlp.248/",
    doi = "10.18653/v1/2023.findings-emnlp.248",
    pages = "3806--3824"
}

@inproceedings{
    jiang2023draft,
    title={Draft, Sketch, and Prove: Guiding Formal Theorem Provers with Informal Proofs},
    author={Albert Qiaochu Jiang and Sean Welleck and Jin Peng Zhou and Timothee Lacroix and Jiacheng Liu and Wenda Li and Mateja Jamnik and Guillaume Lample and Yuhuai Wu},
    booktitle={The Eleventh International Conference on Learning Representations },
    year={2023},
    url={https://openreview.net/forum?id=SMa9EAovKMC}
}

@misc{resi,
      title={ReSI: Recursive Safety Improvement toward Resistant and Resilient AI}, 
      author={Jingnan Zheng and Dongcheng Zhang and Yi Zhang and Ming Zhang and Qiaosheng Zhang and Youbang Sun and An Zhang and Xiangnan He and Tat-Seng Chua and Xia Hu and Bowen Zhou and Chaochao Lu and Xiang Wang},
      year={2026},
      eprint={2610.12233},
      archivePrefix={arXiv},
      primaryClass={cs.CR},
      url={https://arxiv.org/abs/2610.12233}, 
}

@misc{sun2026pogb,
  title     = {Let Solvers Reason and Models Understand: Proof-Obligation-Guided Bridging},
  author    = {Sun, Huan and Huang, Zhongyu and Li, Chaozhuo and Chen, Zhekai and Sun, Youbang and Lu, Chaochao and Hu, Xia},
  year      = {2026},
  month     = {October},
  publisher = {Zenodo},
  doi       = {10.5281/zenodo.23527696},
  url       = {https://doi.org/10.5281/zenodo.23527696},
  note      = {Preprint}
}

@article{wu2025coga,
  title={CoGA: A Collaborative Gray-Box Adversarial Attack for Multimodal Language Models},
  author={Wu, Tong and Lin, Feng and Wang, Gaojian and Liu, Tiantian and Wang, Zhibo and Meng, Weizhi and Liu, Ajian and Ren, Kui},
  journal={IEEE Transactions on Information Forensics and Security},
  volume={21},
  pages={870--885},
  year={2025},
  publisher={IEEE}
}

@inproceedings{
zhou2025robust,
title={Robust Multi-Agent Reinforcement Learning with Stochastic Adversary},
author={Ziyuan Zhou and Guanjun Liu and Mengchu Zhou and Weiran Guo},
booktitle={Forty-second International Conference on Machine Learning},
year={2025},
url={https://openreview.net/forum?id=bnhFueOeav}
}

@article{duan2026last,
  title={The last AI built by humans: Toward genuine recursive self-improvement},
  author={Duan, Yi and Liu, Ying and Tang, Zirui and Chen, Haodong and Zhou, Jun and Liu, Yumou and Xu, Bangrui and Wu, Yukai and Chen, Sidi and Zhou, Yuhan and others},
  journal={arXiv preprint arXiv:2609.11873},
  year={2026}
}

@article{zheng2026dream,
  title={Dream-rsi: Recursive self-improvement through evolving worlds},
  author={Zheng, Tong and Wu, Xidong and Zhang, Zheng and He, Zhankui and Zhang, Chaoyi and Coleman, Benjamin and Wei, Ruoqiao and Bai, Di and Liu, Haolin and Liu, Rui and others},
  journal={arXiv preprint arXiv:2609.14858},
  year={2026}
}

@article{jiang2026survey,
  title={A survey on large language models for code generation},
  author={Jiang, Juyong and Wang, Fan and Shen, Jiasi and Kim, Sungju and Kim, Sunghun},
  journal={ACM Transactions on Software Engineering and Methodology},
  volume={35},
  number={2},
  pages={1--72},
  year={2026},
  publisher={ACM New York, NY}
}

@article{minaee2024large,
  title={Large language models: A survey},
  author={Minaee, Shervin and Mikolov, Tomas and Nikzad, Narjes and Chenaghlu, Meysam and Socher, Richard and Amatriain, Xavier and Gao, Jianfeng},
  journal={arXiv preprint arXiv:2402.06196},
  year={2024}
}

@article{zhang2026march,
  title={MARCH: Scaling Recurrent Memory with Content-Routed State Anchors},
  author={Zhang, Ming and Yang, Kaisen and Yu, Shu and Hua, Ermo and Ding, Ning and Hu, Xia and Zhou, Bowen and Lu, Chaochao and Sun, Youbang},
  journal={arXiv preprint arXiv:2608.12435},
  year={2026}
}

@misc{jin2026architecturallevelsafetycontrolrelocating,
      title={Towards Architectural-Level Safety Control: Relocating Hazardous Knowledge into Removable Experts in Mixture-of-Experts Models}, 
      author={Cheng Jin and Zhongtian Ma and Xiao Wang and Han Qi and Xiaoyu Wen and Xiang Wang and Youbang Sun and Xia Hu and Bowen Zhou and Chaochao Lu and Qiaosheng Zhang},
      year={2026},
      eprint={2610.25791},
      archivePrefix={arXiv},
      primaryClass={cs.CR},
      url={https://arxiv.org/abs/2610.25791}, 
}

@misc{wildteaming2024,
  title = {{WildTeaming at Scale: From In-the-Wild Jailbreaks to (Adversarially) Safer Language Models}},
  author = {Jiang, Liwei and Rao, Kavel and Han, Seungju and Ettinger, Allyson and Brahman, Faeze and Kumar, Sachin and Mireshghallah, Niloofar and Lu, Ximing and Sap, Maarten and Choi, Yejin and Dziri, Nouha},
  year = {2024},
  eprint = {2406.18510},
  archivePrefix = {arXiv},
  url = {https://arxiv.org/abs/2406.18510},
}

@misc{harmbench2024,
  title = {{HarmBench: A Standardized Evaluation Framework for Automated Red Teaming and Robust Refusal}},
  author = {Mazeika, Mantas and Phan, Long and Yin, Xuwang and Zou, Andy and Wang, Zifan and Mu, Norman and Sakhaee, Elham and Li, Nathaniel and Basart, Steven and Li, Bo and Forsyth, David and Hendrycks, Dan},
  year = {2024},
  eprint = {2402.04249},
  archivePrefix = {arXiv},
  url = {https://arxiv.org/abs/2402.04249},
}

@misc{hcot2025,
  title = {{H-CoT: Hijacking the Chain-of-Thought Safety Reasoning Mechanism to Jailbreak Large Reasoning Models, Including OpenAI o1/o3, DeepSeek-R1, and Gemini 2.0 Flash Thinking}},
  author = {Kuo, Martin and Zhang, Jianyi and Ding, Aolin and Wang, Qinsi and DiValentin, Louis and Bao, Yujia and Wei, Wei and Li, Hai and Chen, Yiran},
  year = {2025},
  eprint = {2502.12893},
  archivePrefix = {arXiv},
  url = {https://arxiv.org/abs/2502.12893},
}

@misc{xteaming2025,
  title = {{X-Teaming: Multi-Turn Jailbreaks and Defenses with Adaptive Multi-Agents}},
  author = {Rahman, Salman and Jiang, Liwei and Shiffer, James and Liu, Genglin and Issaka, Sheriff and Parvez, Md Rizwan and Palangi, Hamid and Chang, Kai-Wei and Choi, Yejin and Gabriel, Saadia},
  year = {2025},
  eprint = {2504.13203},
  archivePrefix = {arXiv},
  url = {https://arxiv.org/abs/2504.13203},
}

@misc{xstest2023,
  title = {{XSTest: A Test Suite for Identifying Exaggerated Safety Behaviours in Large Language Models}},
  author = {R{\"{o}}ttger, Paul and Kirk, Hannah Rose and Vidgen, Bertie and Attanasio, Giuseppe and Bianchi, Federico and Hovy, Dirk},
  year = {2023},
  eprint = {2308.01263},
  archivePrefix = {arXiv},
  url = {https://arxiv.org/abs/2308.01263},
}

@misc{mmlupro2024,
  title = {{MMLU-Pro: A More Robust and Challenging Multi-Task Language Understanding Benchmark}},
  author = {Wang, Yubo and Ma, Xueguang and Zhang, Ge and Ni, Yuansheng and Chandra, Abhranil and Guo, Shiguang and Ren, Weiming and Arulraj, Aaran and He, Xuan and Jiang, Ziyan and Li, Tianle and Ku, Max and Wang, Kai and Zhuang, Alex and Fan, Rongqi and Yue, Xiang and Chen, Wenhu},
  year = {2024},
  eprint = {2406.01574},
  archivePrefix = {arXiv},
  url = {https://arxiv.org/abs/2406.01574},
}

@misc{ifeval2023,
  title = {{Instruction-Following Evaluation for Large Language Models}},
  author = {Zhou, Jeffrey and Lu, Tianjian and Mishra, Swaroop and Brahma, Siddhartha and Basu, Sujoy and Luan, Yi and Zhou, Denny and Hou, Le},
  year = {2023},
  eprint = {2311.07911},
  archivePrefix = {arXiv},
  url = {https://arxiv.org/abs/2311.07911},
}

@misc{darwin2026,
  title = {{DARWIN: Evolving Jailbreak Adversary and Guardrail for LLM Safety Evaluation and Protection}},
  author = {Qi, Weiwei and Wu, Zefeng and Guo, Zhilin and Zheng, Tianhang and Lu, Chaochao and He, Liang and Qin, Zhan and Ren, Kui},
  year = {2026},
  eprint = {2607.19829},
  archivePrefix = {arXiv},
  url = {https://arxiv.org/abs/2607.19829},
}

@misc{jailbreakskill2026,
  title = {{JailbreakSkill: Scaling Automated Red-Teaming with Reusable and Ever-Evolving Skills}},
  author = {Wen, Xiaoyu and Li, Jiajia and He, Zhida and Yu, Peng and Wang, Chenxu and Qi, Han and Zhou, Ziyuan and Jin, Cheng and Wen, Ying and Xu, Xingcheng and Hu, Shuyue and Zheng, Tianhang and Lu, Chaochao and Zhang, Qiaosheng},
  year = {2026},
  eprint = {2608.16465},
  archivePrefix = {arXiv},
  url = {https://arxiv.org/abs/2608.16465},
}

@misc{a3_2026,
  title = {{A3: An Automated Alignment Agent for Safety Finetuning}},
  author = {Zhang, Jifan and Sleight, Henry and Benton, Joe},
  year = {2026},
  howpublished = {Alignment Science Blog},
  url = {https://alignment.anthropic.com/2026/automated-alignment-agent/},
  note = {Accessed 15 September 2026},
}

@misc{sinternal2026,
  title = {{Internalizing Safety Understanding in Large Reasoning Models via Verification}},
  author = {Zhang, Yi and Chen, Yuxin and Sheng, Leheng and Zhang, Dongcheng and Lu, Chaochao and Wang, Xiang and Zhang, An},
  year = {2026},
  eprint = {2605.08930},
  archivePrefix = {arXiv},
  url = {https://arxiv.org/abs/2605.08930},
}

@misc{alphaalign2025,
  title         = {{AlphaAlign}: Incentivizing Safety Alignment with Extremely Simplified Reinforcement Learning},
  author        = {Yi Zhang and An Zhang and XiuYu Zhang and Leheng Sheng and Yuxin Chen and Zhenkai Liang and Xiang Wang},
  year          = {2025},
  eprint        = {2507.14987},
  archivePrefix = {arXiv},
  primaryClass  = {cs.AI},
  url           = {https://arxiv.org/abs/2507.14987}
}

@article{yu2025dapo,
  title   = {{DAPO}: An Open-Source {LLM} Reinforcement Learning System at Scale},
  author  = {Yu, Qiying and Zhang, Zheng and Zhu, Ruofei and Yuan, Yufeng
             and Zuo, Xiaochen and Yue, Yu and Dai, Weinan and Fan, Tiantian
             and Liu, Gaohong and Liu, Lingjun and Liu, Xin and Lin, Haibin
             and Lin, Zhiqi and Ma, Bole and Sheng, Guangming and Tong, Yuxuan
             and Zhang, Chi and Zhang, Mofan and Zhang, Wang and Zhu, Hang
             and Zhu, Jinhua and Chen, Jiaze and Chen, Jiangjie and Wang, Chengyi
             and Yu, Hongli and Song, Yuxuan and Wei, Xiangpeng and Zhou, Hao
             and Liu, Jingjing and Ma, Wei-Ying and Zhang, Ya-Qin and Yan, Lin
             and Qiao, Mu and Wu, Yonghui and Wang, Mingxuan},
  journal = {arXiv preprint arXiv:2503.14476},
  year    = {2025},
  doi     = {10.48550/arXiv.2503.14476},
  url     = {https://arxiv.org/abs/2503.14476}
}

@misc{lora2021,
  title = {{LoRA: Low-Rank Adaptation of Large Language Models}},
  author = {Hu, Edward J. and Shen, Yelong and Wallis, Phillip and Allen-Zhu, Zeyuan and Li, Yuanzhi and Wang, Shean and Wang, Lu and Chen, Weizhu},
  year = {2021},
  eprint = {2106.09685},
  archivePrefix = {arXiv},
  url = {https://arxiv.org/abs/2106.09685},
}

@misc{r2ai2025,
  title = {{R$^2$AI: Towards Resistant and Resilient AI in an Evolving World}},
  author = {Sun, Youbang and Wang, Xiang and Fu, Jie and Lu, Chaochao and Zhou, Bowen},
  year = {2025},
  eprint = {2509.06786},
  archivePrefix = {arXiv},
  url = {https://arxiv.org/abs/2509.06786}
}

@misc{claudeopus5systemcard2026,
  title = {{System Card: Claude Opus 5}},
  author = {{Anthropic}},
  year = {2026},
  howpublished = {System card},
  url = {https://www.anthropic.com/claude-opus-5-system-card},
  note = {Accessed 20 September 2026}
}

@misc{kimik3report2026,
  title = {{Kimi K3: Open Frontier Intelligence}},
  author = {{Moonshot AI}},
  year = {2026},
  howpublished = {Technical report},
  url = {https://github.com/MoonshotAI/Kimi-K3/blob/main/k3_tech_report.pdf}
}

@misc{glm53model2026,
  title = {{GLM-5.3}},
  author = {{Z.ai}},
  year = {2026},
  howpublished = {Official model card},
  url = {https://huggingface.co/zai-org/GLM-5.3},
  note = {Accessed 20 September 2026}
}

@misc{deepseekmath2024,
  title = {{DeepSeekMath: Pushing the Limits of Mathematical Reasoning in Open Language Models}},
  author = {Shao, Zhihong and Wang, Peiyi and Zhu, Qihao and Xu, Runxin and Song, Junxiao and Bi, Xiao and Zhang, Haowei and Zhang, Mingchuan and Li, Y. K. and Wu, Y. and Guo, Daya},
  year = {2024},
  eprint = {2402.03300},
  archivePrefix = {arXiv},
  url = {https://arxiv.org/abs/2402.03300},
}
\endgroup

\clearpage
\appendix
\section{Appendix}

\subsection{Running Example}
\label{app:exam}

In this section, we illustrate our proposed method through a simplified example from the FOLIO dataset~\citep{han2024folio}.
The NL premise set $\symbPre = \{\symbpre_1; \symbpre_2; \symbpre_3\}$ states: “Reading a book gives a person knowledge; gaining knowledge makes that person smarter; and Harry read \emph{Walden}.”
The NL query $\symbhyp$ asks whether the statement “Harry is smarter than before” can be concluded from these premises.
Let the constants $h$, $w$, and $k$ denote Harry, \emph{Walden}, and knowledge, respectively.
We first perform a grounding autoformalization ${\symbAutofm}$ to map $\symbqa = (\symbPre,\symbhyp)$ to
\begin{equation*}
\begin{aligned}
\symbPF = \{\;
    &\symbpf_1:\ \forall x_1\forall y.\ 
        \bigl(\mathsf{Book}(y)\land\mathsf{Read}(x_1,y)\bigr)
        \rightarrow\mathsf{Gain}(x_1,k),\\
    &\symbpf_2:\ \forall x_2.\ 
        \mathsf{Gain}(x_2,k)\rightarrow\mathsf{Smarter}(x_2),\\
    &\symbpf_3:\ \mathsf{Book}(w)\land\mathsf{Read}(h,w)\;\},
\qquad
\symbHF = \mathsf{SmarterThanBefore}(h),
\end{aligned}
\label{eq:running-example}
\end{equation*}
and the grounding table~$\symbSymTab$ shown in Table~\ref{tab:gt}.
During proof-time inference, the solver can therefore access all entries in $\symbSymTab$.
In particular, type information $\operatorname{type}_{\symbSymTab}(\symbterm)$ constrains which instantiations are admissible, whereas NL information $\operatorname{NL}_{\symbSymTab}(\symbterm)$ guides the search for semantic bridges when exact symbolic matching fails.

\begin{table}[H]\small
    \renewcommand{\arraystretch}{1.25}
    \centering
    \caption{The grounding table $\symbSymTab$ of the running example.}
    \label{tab:gt}
    \begin{tabular}{lll}
    \toprule
    \makecell[l]{\textbf{Formal element}\\$\symbterm$}
    & \makecell[l]{\textbf{Type information}\\$\operatorname{type}_{\symbSymTab}(\symbterm)$}
    & \makecell[l]{\textbf{Natural-language information}\\$\operatorname{NL}_{\symbSymTab}(\symbterm)$} \\
    \midrule
    $h$     & $\mathsf{Person}$ & “Harry”\\
    $w$     & $\mathsf{Book}$   & “\emph{Walden}”\\
    $k$     & $\mathsf{Object}$ & “knowledge”\\
    $x_1$   & $\mathsf{Person}$ & “a person”\\
    $x_2$   & $\mathsf{Person}$ & “a person”\\
    $y$     & $\mathsf{Book}$   & “a book”\\
    $\mathsf{Book}(\{a_1\})$
            & $\mathsf{Book}$
            & “$\{a_1\}$ is a book”\\
    $\mathsf{Read}(\{a_1\},\{a_2\})$
            & $\mathsf{Person}, \mathsf{Book}$
            & “$\{a_1\}$ reads $\{a_2\}$”\\
    $\mathsf{Gain}(\{a_1\},k)$
            & $\mathsf{Person}$
            & “$\{a_1\}$ gains knowledge”\\
    $\mathsf{Smarter}(\{a_1\})$
            & $\mathsf{Person}$ & “$\{a_1\}$ is smarter”\\
    $\mathsf{SmarterThanBefore}(\{a_1\})$
            & $\mathsf{Person}$ & “$\{a_1\}$ is smarter than before”\\
    \bottomrule
    \end{tabular}
\end{table}

After completing the grounding autoformalization process described above, we perform the proof-obligation-guided search.
Initialize the proof obligation as $\phi = \symbHF = \mathsf{SmarterThanBefore}(h)$.
Algorithm~\ref{alg:grounding-retrieval} instantiates $\symbpf_1$ and $\symbpf_2$ using the substitutions $x_1 \mapsto h$, $x_2 \mapsto h$, and $y \mapsto w$, yielding $\symbinst_1$ and $\symbinst_2$.
Type compatibility permits these substitutions while ruling out alternatives such as $x_1 \mapsto w$, $x_1 \mapsto k$, $x_2 \mapsto w$, $x_2 \mapsto k$, $y \mapsto h$, and $y \mapsto k$.
Algorithm~\ref{alg:grounding-retrieval} considers ${\symbInst} =\{\symbinst_2\}$ and obtains
\begin{equation*}
    \Call{Lit}{\symbinst_2} = \{ \neg \mathsf{Gain}(h,k), \mathsf{Smarter}(h) \},
\end{equation*}
It then selects ${\mathbb{S}_{\mathrm{lit}}} = \{\mathsf{Smarter}(h)\} \subseteq \Call{Lit}{\symbinst_2}$, and constructs the candidate semantic support $\Sigma = \mathsf{Smarter}(h)$ and the candidate semantic bridge
\begin{align*}
    \symbbridge
    &= \Sigma \rightarrow \phi \\
    &= \mathsf{Smarter}(h) \rightarrow \mathsf{SmarterThanBefore}(h).
\end{align*}
Logical abduction then derives the residual obligation $\symbobligation = \mathsf{Gain}(h,k)$.
Once the LLM-based $\operatorname{Gate}$ accepts $\symbbridge$ based on its NL interpretation that “Harry is smarter” entails “Harry is smarter than before,” Algorithm~\ref{alg:direct-search} proceeds to discharge the residual obligations:
\begin{align*}
    \mathsf{SmarterThanBefore}(h)
    &\xrightarrow[\symbbridge]{\symbpf_2[x_2 \mapsto h]} \{\mathsf{Gain}(h,k)\} \\
    &\xrightarrow{\symbpf_1[x_1 \mapsto h, y \mapsto w]} \{\mathsf{Book}(w), \mathsf{Read}(h,w)\} \\
    &\xrightarrow{\symbpf_3} \checkmark.
\end{align*}
Here, each arrow represents the residual obligations produced in a single round of derivation, with bridging applied as needed.
Finally, premise $\symbpf_3$ discharges the two remaining obligations.
In this example, only the displayed bridge $\symbbridge$ requires LLM evaluation throughout the derivation.

\subsection{Proof of Soundness}
\label{app:recursive-soundness}

\newtheorem{searchlemma}{Lemma}

\begin{searchlemma}[Existential decomposition]
\label{lem:search-existential}
For $\phi=\exists x.\,\psi$, let $\symbBridge$ be the bridge set
returned by a recursive call for an admissible instance
$\psi[\theta]\in\textnormal{\textsc{Inst}}(\psi,\symbSymTab)$.
If ${\symbSolver(\symbPF\cup\symbBridge,\psi[\theta])=\textnormal{\textsc{Proved}}}$,
then ${\symbSolver(\symbPF\cup\symbBridge,\phi)=\textnormal{\textsc{Proved}}}$.
\end{searchlemma}
\begin{proof}
Apply existential introduction with the witness $\theta$'s bridge set.
\end{proof}

\begin{searchlemma}[Universal decomposition]
\label{lem:search-universal}
For $\phi=\forall x.\,\psi$, introduce a fresh eigenvariable
${c_{\mathrm{fresh}}}$ and let $\symbBridge$ be the bridge set returned for
$\psi[x\mapsto {c_{\mathrm{fresh}}}]$.  Each bridge involving ${c_{\mathrm{fresh}}}$ is
universally quantified over the corresponding variable and accepted by Gate
with that universal scope.  If
${\symbSolver(\symbPF\cup\symbBridge,\psi[x\mapsto {c_{\mathrm{fresh}}}])=\textnormal{\textsc{Proved}}}$,
then ${\symbSolver(\symbPF\cup\symbBridge,\phi)=\textnormal{\textsc{Proved}}}$.
\end{searchlemma}
\begin{proof}
The relevant bridges are asserted for every admissible assignment to the
variable represented by ${c_{\mathrm{fresh}}}$.  Freshness and universal closure ensure
that the assumptions do not depend on a particular value of ${c_{\mathrm{fresh}}}$.
Universal introduction therefore extends the subproof of
$\psi[x\mapsto {c_{\mathrm{fresh}}}]$ to a proof of $\forall x.\,\psi$, using
the same bridge set.
\end{proof}

\begin{searchlemma}[Conjunctive decomposition]
\label{lem:search-conjunction}
For $\phi=\phi_1\land\phi_2$, let $\symbBridge_1$ and
$\symbBridge_2$ be the bridge sets returned by the two recursive calls.
If ${\symbSolver(\symbPF\cup\symbBridge_i,\phi_i)=\textnormal{\textsc{Proved}}}$
for $i=1,2$, then
their merged set $\symbBridge=\symbBridge_1\cup\symbBridge_2$ satisfies
${\symbSolver(\symbPF\cup\symbBridge,\phi)=\textnormal{\textsc{Proved}}}$.
\end{searchlemma}
\begin{proof}
Combine the two subproofs by conjunction introduction. 
\end{proof}

\begin{searchlemma}[Disjunctive decomposition]
\label{lem:search-disjunction}
For $\phi=\phi_1\lor\phi_2$, select any successful branch
$i\in\{1,2\}$ with returned bridge set $\symbBridge_i$.
If ${\symbSolver(\symbPF\cup\symbBridge_i,\phi_i)=\textnormal{\textsc{Proved}}}$,
then the selected set $\symbBridge=\symbBridge_i$ satisfies
${\symbSolver(\symbPF\cup\symbBridge,\phi)=\textnormal{\textsc{Proved}}}$.
\end{searchlemma}
\begin{proof}
Apply disjunction introduction to the selected subproof, retaining only its
bridge set.
\end{proof}

\begin{searchlemma}[Bridge identification]
\label{lem:search-bridge}
If Algorithm~\ref{alg:grounding-retrieval} yields
$(\symbBridge,\symbobligation)$ for $\phi$, 
${\symbSolver(\symbPF\cup\symbBridge_{\symbobligation},\symbobligation)
=\textnormal{\textsc{Proved}}}$ with a Gate-accepted bridge set
$\symbBridge_{\symbobligation}$, then
\begin{equation*}
\symbSolver(\symbPF\cup\symbBridge\cup\symbBridge_{\symbobligation},\phi)
=\textnormal{\textsc{Proved}}
\quad\land\quad
\operatorname{Gate}(\symbBridge\cup\symbBridge_{\symbobligation},\symbSymTab)
=\textnormal{\textsc{Accepted}}.
\end{equation*}
\end{searchlemma}
\begin{proof}
The yield statement requires Gate acceptance.  The instantiated premise set
$\symbInst$ follows from $\symbPF$.  By the abduction condition in
Section~\ref{sec:proof-completion}, the instantiated premises and the residual subproof
jointly establish the selected support $\Sigma$. 
\end{proof}

\begin{theorem}[Soundness of successful search]
\label{thm:recursive-search-soundness}
If Algorithm~\ref{alg:direct-search} returns a bridge set $\symbBridge$
for a goal $\symbgoal$, then $\symbBridge$ satisfies both criteria in  Equation~\eqref{eq:accepted-completion}:
\begin{equation*}
\symbSolver(\symbPF\cup\symbBridge, \symbgoal) = \textsc{Proved}
\quad\land\quad
\operatorname{Gate}(\symbBridge, \symbSymTab) = \textsc{Accepted}.
\end{equation*}
\end{theorem}

\begin{proof}
We prove by induction on the successful recursive derivation that each successful call
$\Call{Prove}{\phi,\symbAncestor,\symbSymTab}$ returns a Gate-accepted
bridge set under which the solver can discharge $\phi$.

\emph{Base case.}
The solver directly proves $\phi$ from $\symbPF$.  The returned
bridge set is empty, so Gate acceptance is vacuous.

\emph{Inductive step.}
Assume the claim for the successful recursive subcalls.  Normalization
preserves equivalence, so it suffices to discharge the normalized obligation.
Its four logical decomposition cases are as follows.
\begin{enumerate}[label=(\roman*),leftmargin=2.5em,itemsep=0.35em,parsep=0pt,topsep=0.4em]
\item $\phi=\exists x.\,\psi$.
A recursive call succeeds for an admissible instance $\psi[\theta]$.
The induction hypothesis establishes that instance with accepted bridges,
and Lemma~\ref{lem:search-existential} discharges $\phi$.

\item $\phi=\forall x.\,\psi$.
The recursive call succeeds for $\psi[x\mapsto {c_{\mathrm{fresh}}}]$, with
${c_{\mathrm{fresh}}}$ recorded as arbitrary in the extended grounding table.
The induction hypothesis establishes the body with accepted bridges;
those involving ${c_{\mathrm{fresh}}}$ retain their universal scope.
Lemma~\ref{lem:search-universal} discharges $\phi$.

\item $\phi=\phi_1\land\phi_2$.
Both recursive calls must succeed.  By the induction hypothesis, their
returned sets $\symbBridge_1$ and $\symbBridge_2$ contain accepted
bridges and establish $\phi_1$ and $\phi_2$, respectively.
Lemma~\ref{lem:search-conjunction} discharges $\phi$ under the merged set.

\item $\phi=\phi_1\lor\phi_2$.
At least one recursive call succeeds.  The induction hypothesis establishes
the disjunct of any selected successful branch with its accepted bridge set
$\symbBridge_i$.  Lemma~\ref{lem:search-disjunction} discharges $\phi$
with $\symbBridge=\symbBridge_i$.

\end{enumerate}

If no logical decomposition applies, Algorithm~\ref{alg:grounding-retrieval}
yields $(\symbBridge,\symbobligation)$, and Algorithm~\ref{alg:direct-search}
recursively proves $\symbobligation$.  The induction hypothesis supplies
its proof and accepted bridge set $\symbBridge_{\symbobligation}$.
Lemma~\ref{lem:search-bridge} then discharges $\phi$ under
$\symbBridge\cup\symbBridge_{\symbobligation}$, whose bridges are
accepted by that lemma and the induction hypothesis.
Thus every successful return supplies a symbolic proof and accepted bridges.
Taking $\phi=\symbgoal$ gives Equation~\eqref{eq:accepted-completion}.
\end{proof}


\end{document}